\documentclass[11pt]{article} 

\usepackage{fullpage}
\usepackage{amsmath}
\usepackage{amssymb}
\usepackage{amsthm}
\usepackage{amsfonts}
\usepackage{mathrsfs}
\usepackage{comment}

\usepackage[numbers]{natbib}
\usepackage{color}
\usepackage{graphicx}
\usepackage{enumerate}
\usepackage{epstopdf}
\usepackage{epsfig}
\usepackage{appendix}
\usepackage{comment}
\usepackage{fixltx2e}

\usepackage{authblk}
\usepackage{subcaption}

\newtheorem{theorem}{Theorem}
\newtheorem{lemma}{Lemma}

\newtheorem{corollary}{Corollary}
\newtheorem{definition}{Definition}

\newtheorem{innermech}{Mechanism}
\newenvironment{mechanism}[1]
  {\renewcommand\theinnermech{#1}\innermech}
  {\endinnermech}

\begin{document}
\title{Facility Location with Double-peaked Preferences\thanks{Aris Filos-Ratsikas and Jie Zhang were supported by the Sino-Danish Center for the Theory of Interactive Computation, funded by the Danish National Research Foundation and the National Science Foundation of China (under the grant 61061130540), and by the Center for research in the Foundations of Electronic Markets (CFEM), supported by the Danish Strategic Research Council. Jie Zhang was also supported by ERC Advanced Grant 321171 (ALGAME). Minming Li was partly supported by a grant from the Research Grants Council of the Hong Kong Special Administrative Region, China [Project No. CityU 117913]. Qiang Zhang was supported by FET IP project MULTIPEX 317532.
}}

\author[$\dag$]{Aris Filos-Ratsikas}
\author[$\S$]{Minming Li}
\author[$\star$]{Jie Zhang}
\author[$\ddagger$]{Qiang Zhang}

\affil[$\dag$]{\small{Department of Computer Science, Aarhus University, Denmark.}}
\affil[$\S$]{\small{Department of Computer Science, City University of Hong Kong, HK.}}
\affil[$\star$]{\small{Department of Computer Science, University of Oxford, UK.}}
\affil[$\ddagger$]{\small{Institute of Informatics, University of Warsaw, Poland.}}

\date{}
%
%

\maketitle
\begin{abstract}
We study the problem of locating a single \emph{facility} on a real line based on the reports of self-interested agents, when agents have \emph{double-peaked preferences}, with the peaks being on opposite sides of their locations. We observe that double-peaked preferences capture real-life scenarios and thus complement the well-studied notion of single-peaked preferences. We mainly focus on the case where peaks are equidistant from the agents' locations and discuss how our results extend to more general settings.  We show that most of the results for single-peaked preferences do not directly apply to this setting; this makes the problem essentially more challenging. As our main contribution, we present a simple truthful-in-expectation mechanism that achieves an approximation ratio of $1+b/c$ for both the social and the maximum cost, where $b$ is the distance of the agent from the peak and $c$ is the minimum cost of an agent. For the latter case, we provide a $3/2$ lower bound on the approximation ratio of any truthful-in-expectation mechanism. We also study deterministic mechanisms under some natural conditions, proving lower bounds and approximation guarantees. We prove that among a large class of reasonable mechanisms, there is no deterministic mechanism that outperforms our truthful-in-expectation mechanism.
\end{abstract}

\section{Introduction}
We study the problem of locating a single \emph{facility} on a real line, based on the input provided by selfish agents who wish to minimize their costs. Each agent has a \emph{location} $x_i \in \mathbb{R}$ which is her private information and is asked to report it to some central authority, which then decides where to locate the facility, aiming to optimize some function of the agents' reported locations. This model corresponds to problems such as finding the ideal location for building a primary school or a bus stop along a street, so that the total distance of all agents' houses from the location is minimized, or so that no agent's house will lie too far away from that location.

In our setting, we assume that agents have \emph{double-peaked} preferences, i.e. we assume that each agent $i$ has two unique most preferred points or \emph{peaks}, located at some distances from $x_i$ on opposite sides, where her cost is minimum. Traditionally, preferences in facility location problems are assumed to be \emph{single-peaked}, i.e. each agent's location is her most preferred point on the line and her cost increases linearly (at the same rate) to the left and the right of that peak. Sometimes however, single-peaked preferences do not model real-life scenarios accurately.
Take for instance the example mentioned above, where the government plans to build a primary school on a street. An agent with single-peaked preferences would definitely want the school built next to her house, so that she wouldn't have to drive her children there everyday. However, it is quite possible that she is also not very keen on the inevitable drawbacks of having a primary school next to her house either, like unpleasant noise or trouble with parking. On the other hand, a five-minute walking distance is sufficiently far for those problems to no longer be a factor but also sufficiently close for her children to be able to walk to school. There are two such positions, (symmetrically) in each direction, and those would be her two peaks.  

Our primary objective is to explore double-peaked preferences in facility location settings similar to the ones studied extensively for single-peaked preferences throughout the years \cite{procaccia2009approximate,schummer2002strategy,lu2009tighter,lu2010asymptotically,alon2010strategyproof,fotakis2010winner,escoffier2011strategy,fotakis2012power,dokow2012mechanism,feldman2013strategyproof}. For that reason, following the literature we assume that the cost functions are the same for all agents and that the cost increases linearly, at the same rate, as the output moves away from the peaks. The straightforward extension to the double-peaked case is piecewise-linear cost functions, with the same slope in all intervals, which gives rise to the natural model of \emph{symmetric} agents, i.e. the peaks are equidistant from the agent's location. Note that this symmetry is completely analogous to the single-peaked case (for facility location problems, e.g. see \cite{procaccia2009approximate}), where agents have exactly the same cost on two points equidistant from their peaks. Our lower bounds and impossibility results naturally extend to non-symmetric settings, but some of our mechanisms do not. We discuss those extensions in Section \ref{generalizations}. 

Our model also applies to more general spaces, beyond the real line. One can imagine for instance that the goal is to build a facility on the plane where for the same reasons, agents would like the facility to be built at some distance from their location, in \emph{every direction}. This translates to an agent having infinitely many peaks, located on a circle centered around her location. In that case of course, we would no longer refer to agents' preferences as double-peaked but the underyling idea is similar to the one presented in this paper. We do not explore such extensions here; we leave that for future work.  

Agents are self-interested entities that wish to minimize their costs. We are interested in mechanisms that ensure that agents are not incentivized to report anything but their actual locations, namely \emph{strategyproof} mechanisms. We are also interested in \emph{group strategyproof} mechanisms, i.e., mechanisms that are resistant to manipulation by coalitions of agents. Moreover, we want those mechanisms to achieve some good performance guarantees, with respect to our goals. If our objective is to minimize the sum of the agent's costs, known as the \emph{social cost}, then we are looking for strategyproof mechanisms that achieve a social cost as close as possible to that of the optimal mechanism, which need not be strategyproof. The prominent measure of performance for mechanisms in computer science literature is the approximation ratio \cite{dughmi2010truthful,guo2010strategyproof,ashlagi2010mix,caragiannis2011improved}, i.e., the worst possible ratio of the social cost achieved by the mechanism over the minimum social cost over all instances of the problem. The same holds if our objective is to minimize the \emph{maximum cost} of any agent. In the case of \emph{randomized mechanisms}, i.e., mechanisms that output a probability distribution over points in $\mathbb{R}$, instead of a single point, as a weaker strategyproofness constraint, we require \emph{truthfulness-in-expectation}, i.e., a guarantee that no agent can reduce her expected cost from misreporting.

\subsection{Double-peaked preferences in practice}

Single-peaked preferences were introduced in \cite{black1986theory} as a way to avoid \emph{Condorcet cycles} in majority elections. Moulin \cite{moulin1980strategy} characterized the class of strategyproof mechanisms in this setting, proving that \emph{median voter schemes} are essentially the only strategyproof mechanisms for agents with single-peaked preferences. 
Double-peaked preferences have been mentioned in social choice literature (e.g. see \cite{cooper1999thestrategic}), to describe settings where preferences are not single-peaked, voting cycles do exist and majority elections are not possible. In broader social choice settings, they can be used to model situations where e.g. a left-wing party might prefer a more conservative but quite effective policy to a more liberal but ineffective one on a left-to-right political axis. In fact, Egan \cite{egan2013something} provides a detailed discussion on double-peaked preferences in political decisions. He uses a 1964-1970 survey about which course of action the United States should take with regard to the Vietnam war as an example where the status quo (keep U.S, troops in Vietnam but try to terminate the war) was ranked last by a considerable fraction of the population when compared to a left-wing policy (pull out entirely) or a right-wing policy (take a stronger stand). This demonstrates that in a scenario where the standard approach would be to assume that preferences are single-peaked, preferences can instead be double-peaked. Egan provides additional evidence for the occurence of double-peaked preferences supported by experimental results based on surveys on the U.S. population, for many different problems (education, health care, illegal immigration treatment, foreign oil treatment e.t.c.). More examples of double-peaked preferences in real-life scenarios are presented in \cite{rosen2004public}. The related work demonstrates that although they might not be as popular as their single-peaked counterpart, double-peaked preferences do have applications in settings more general than the street example described earlier. On the other hand, the primary focus of this paper is to study double-peaked preferences on facility location settings and therefore the modelling assumptions follow the ones of the facility location literature.

\subsection{Our results}

Our main contribution is a truthful-in-expectation mechanism (\ref{leftrightmedian}) that achieves an approximation ratio of $1+b/c$ for the social cost and $\max\{1+b/c,2\}$ for the maximum cost, where $b$ is the distance between an agent's location and her peak and $c$ is her minimum cost. We also prove that no truthful-in-expectation mechanism can do better than a $3/2$ approximation for the maximum cost proving that at least for the natural special case where $b=c$, Mechanism \ref{leftrightmedian} is not far from the best possible. For deterministic mechanisms, we prove that no mechanism in a wide natural class of strategyproof mechanisms can achieve an approximation ratio better than $1+b/c$ for the social cost and $1+2b/c$ for the maximum cost and hence cannot outperform Mechanism \ref{leftrightmedian}. To prove this, we first characterize the class of strategyproof, anonymous and position invariant mechanisms for two agents by a single mechanism (\ref{mech}). Intuitively, anonymity requires that all agents are handled equally by the mechanism while position invariance essentially requires that if we shift an instance by a constant, the location of the facility should be shifted by the same constant as well. This is a quite natural condition and can be interpreted as a guarantee that the facility will be located \emph{relatively} to the reports of the agents and independently of the underlying network (e.g. the street). 

We prove that the approximation ratio of Mechanism \ref{mech} for the social cost is $\Theta(n)$, where $n$ is the number of agents and conjecture that no deterministic strategyproof mechanism can achieve a constant approximation ratio in this case. For the maximum cost, the ratio of Mechanism \ref{mech} is $\max\{1+2b/c,3\}$ which means that the mechanism is actually the best in the natural class of anonymous and position invariant mechanisms. For any deterministic strategyproof mechanism, we prove a lower bound of $2$ on the approximation ratio, proving that at least for the natural case of $b=c$, Mechanism \ref{mech} is also not far from optimal.
Finally, we prove an impossibility result; there is no group strategyproof, anonymous and position invariant mechanism for the problem. This is in constrast with the single-peaked preference setting, where there is a large class of group strategyproof mechanisms that satisfy those properties. Our results are summarized in Table \ref{resultstable}.

\section{Preliminaries}\label{preliminaries}
Let $N=\{1,2,\ldots,n\}$ be a set of {\em agents}. We consider the case where agents are located on a line, i.e., each agent $i \in N$ has a location $x_i \in \mathbb R$. We will occasionally use $x_i$ to refer to both the position of agent $i$ and the agent herself. We will call the collection $\mathbf x = \langle x_1, \dots, x_n\rangle $ a {\em location profile} or an {\em instance}. 

A {\em deterministic mechanism} is a function $f: \mathbb R^n \mapsto \mathbb R$  that maps a given location profile to a point in $\mathbb{R}$, the location of the {\em facility}. We assume that agents have  {\em double-peaked preferences}, symmetric with respect to the origin. We discuss how our results extend to non-symmetric agents in Section \ref{generalizations}. Given any instance $\mathbf x$ and a location $y \in \mathbb{R}$, the cost of agent $i$ is
\begin{align*}
\mathrm{cost}(y, x_i) =
\begin{cases}
c+|x_i - b - y| & \text{if}\ y \le x_i \\
c+|x_i + b - y| & \text{if}\ y > x_i
\end{cases}
\end{align*}
where $c$ and $b$ are positive constants. We will say that $y$ \emph{admits} a cost of $\mathrm{cost}(y, x_i)$ for agent $i$ on instance 
$\mathbf{x}$. For a mechanism that outputs $f(\mathbf{x})$ on instance $\mathbf{x}$, the cost of agent $i$ is $\mathrm{cost}(f(\mathbf{x}),x_i)$. 
Intuitively, each agent has two most favorable locations, i.e., $x_i - b$ and $x_i + b$, which we refer to as the {\em peaks} of agent $i$. Note that these peaks are actually the troughs of the curve of the cost function, but much like most related work, we refer to them as peaks. The parameter $c>0$ is the minimum cost incurred to an agent when the facility is built on one of her peaks.\footnote{It is not hard to see by our results that if we let an agent's cost be zero on her peaks, then in very general settings, no determnistic strategyproof mechanism can guarantee a finite approximation ratio.} Note that the special case, where $b=c$ corresponds to the natural setting where the incurred minimum cost of an agent is interpreted as the distance she needs to cover to actually reach the facility. This case is particularly appealing, since the bounds we obtain are clean numbers, independent of $b$ and $c$. The bounds for the natural case can be obtained directly by letting $b=c$ in all of our results. 

A {\em randomized mechanism} is a function $f: \mathbb R^n \mapsto \Delta (\mathbb R)$, where $\Delta (\mathbb R)$ is the set of probability distributions over $\mathbb R$. It maps a given location profile to probabilistically selected locations of the facility. The expected cost of agent $i$ is $\mathbb E _{y \sim \mathcal D} \left[\mathrm{cost}(y, x_i) \right]$, where $\mathcal D$ is the probability distribution of the mechanism outputs.

We will call a deterministic mechanism $f$ {\em strategyproof} if no agent would benefit by misreporting her location, regardless of the locations of the other agents. This means that for every $\mathbf x\in \mathbb R^n$, every $i\in N$ and every $x'_i\in \mathbb R$, $\mathrm{cost}(f(\mathbf x),x_i)\le \mathrm{cost}(f(x'_i, \mathbf {x}_{-i}),x_i)$,  where $\mathbf {x}_{-i}=\langle x_1,\dots,x_{i-1},x_{i+1},\dots,x_n \rangle$. A mechanism is {\em truthful-in-expectation} if it guarantees that every agent always minimizes her expected cost by reporting her location truthfully. Throughout the paper we will use the term \emph{strategyproofness} when refering to deterministic mechanisms and the term \emph{truthfulness} when refering to randomized mechanisms.

A mechanism is {\em group strategyproof} if there is no coalition of agents, who by jointly misreporting their locations, affect the outcome in a way such that the cost of none of them increases and the cost of at least one of them strictly decreases. In other words, there is no $S \subseteq N$ such that for some misreports $x_S'$ of agents in $S$ and some reports $\mathbf{x}_{-S}$ of agents in $N\backslash S$,  $\mathrm{cost}(f(x_S', \mathbf {x}_{-S}),x_i) \leq \mathrm{cost}(f(\mathbf x),x_i)$
for all $i \in S$, and $\mathrm{cost}(f(x_S', \mathbf {x}_{-S}),x_j) < \mathrm{cost}(f(\mathbf x),x_j)$ for at least one $j \in S$.

Given an instance $\mathbf x$ and a location $y \in \mathbb{R}$, the {\em social cost} and the {\em maximum cost} of $y$ are defined respectively as: 
\[
SC_y(\mathbf{x}) = \sum^n_{i=1} \mathrm{cost} (y, x_i)\ \ \  \textrm{,} \ \  \  MC_y(\mathbf{x}) = \max_{i \in N} \mathrm{cost} (y, x_i).
\]

We will say that $y$ \emph{admits} a social cost of $SC_y(\mathbf{x})$ or a maximum cost of $MC_y(\mathbf{x})$. We will call $y \in \mathbb{R}$ an \emph{optimal location} (for the social cost), if $y \in \arg\min_y SC_y(\mathbf{x})$. The definition for the maximum cost is analogous. Let $SC_{\mathrm{opt}}(\mathbf x)$ and $MC_{\mathrm{opt}}(\mathbf x)$ denote the social cost and the maximum cost of an optimal location respectively, on instance $\mathbf{x}$. For a mechanism $f$ that outputs $f(\mathbf{x})$ on instance $\mathbf{x}$, we will call $SC_\mathrm{f(\mathbf{x})}(\mathbf x)$ the social cost of the mechanism and we will denote it by $SC_\mathrm{f}(\mathbf x)$; and analogously for the maximum cost. 

We are interested in strategyproof mechanisms that perform well with respect to the goal of minimizing either the social cost or the maximum cost. We measure the performance of the mechanism by comparing the social/maximum cost it achieves with the optimal social/maximum cost, on any instance $\mathbf x$.

The approximation ratio of mechanism $f$, with respect to the social cost,  is given by
\begin{align*}
r=\sup_{\mathbf{x}} \frac{SC_\mathrm{f}(\mathbf x)}{SC_{\mathrm{opt}}(\mathbf x)}.
\end{align*}
The approximation ratio of mechanism $f$, with respect to maximum cost,  is defined similarly.

For randomized mechanisms, the definitions are similar and the approximation ratio is calculated with respect to the expected social or maximum cost, i.e., the expected sum of costs of all agents and expected maximum cost of any agent, respectively.

Finally we consider some properties which are quite natural and are satisfied by many  mechanisms (including the optimal mechanism). A mechanism $f$ is {\em anonymous}, if for every location profile $\mathbf x$ and every permutation $\pi$ of the agents, $f(x_1, \ldots, x_n) = f(x_{\pi(1)}, \ldots, x_{\pi(n)})$. We say that a mechanism $f$ is {\em onto}, if for every point $y \in \mathbb R$ on the line, there exists a location profile $\mathbf x$ such that $f(\mathbf x) = y $. Without loss of generality, for anonymous mechanisms, we can assume $x_1\le  \dots \le x_n$. 

A property that requires special mention is that of \emph{position invariance}, which is a very natural property as discussed in the introduction. This property was independently defined by \cite{feigenbaum2013approximately} where it was referred to as \emph{shift invariance}. One can view position invariance as an analogue to \emph{neutrality} in problems like the one studied here, where there is a continuum of outcomes instead of a finite set.

\begin{definition}
	\label{dfn-pi}
	A mechanism $f$ satisfies {\em position invariance}, if for all location profiles $\mathbf{x}=\langle x_1,...,x_n \rangle$ and $t \in \mathbb{R}$, it holds $f(x_1+t, x_2+t, \ldots, x_n+t) = f(\mathbf x) + t$. In this case, we will call such a mechanism \emph{position invariant}. We will refer to instances $\mathbf x$ and $\langle x_1+t, x_2+t, \ldots, x_n+t \rangle$ as {\em position equivalent}.
\end{definition}

Note that position invariance implies the onto condition. Indeed, for any location profile $\mathbf{x}$, with $f(\mathbf{x})=y$, we have $f(x_1+t, x_2+t, \ldots, x_n+t)=y'=y+t$ for any $t \in \mathbb{R}$, so every point $y' \in \mathbb{R}$ is a potential output of the mechanism.

\section{A truthful-in-expectation mechanism}

We start the exposition of our results with our main contribution, a truthful-in-expectation mechanism that achieves an approximation ratio of $1+b/c$ for the social cost and $\max\{1+b/c,2\}$ for the maximum cost.

\begin{mechanism}{M1}\label{leftrightmedian}
	Given any instance $\mathbf{x}=\langle x_1,...,x_n\rangle$, find the median agent $x_{m}=\mathrm{median}(x_1,...,x_n)$, breaking ties in favor of the agent with the smallest index. Output $f(\mathbf{x})=x_m-b$ with probability $\frac{1}{2}$ and $f(\mathbf{x})=x_m+b$ with probability $\frac{1}{2}$.
\end{mechanism}

\begin{theorem}
	Mechanism \ref{leftrightmedian} is truthful-in-expectation.
\end{theorem}

\begin{proof}
	First, note that the median agent does not have an incentive to deviate, since her expected cost is already minimum, neither does any agent $i$ for which $x_i=x_m$. Hence, for the deviating agent $i$ it must be either $x_i < x_m$ or $x_i> x_m$. We consider three cases when $x_i<x_m$. The proof for the case $x_i>x_m$ is symmetric. Observe that for agent $i$ to be able to move the position of the facility, she has to report $x_i'\geq x_m$ and change the identity of the median agent. Let $x'_m$ be the median agent in the new instance $\langle x'_i,x_{-i} \rangle$, after agent $i$'s deviation. If $x'_m=x_m$, then obviously agent $x_i$ does not gain from deviating, so we will assume that $x'_m>x_m$. 
	
	\textbf{Case 1:} $x_i + b \leq x_m-b$ (symmetrically $x_i-b \geq x_m+b$).
	
	In this case, the cost of agent $i$ is calculated with respect to $x_i+b$ for both possible outcomes of the mechanism. Since $x'_m-b > x_m-b$ and $x'_m+b > x_m+b$, it holds that $|(x_i+b)-(x'_m-b)| > |(x_i+b)-(x_m-b)|$ and $|(x_i+b)-(x'_m+b|) > |(x_i+b)-(x_m+b)|$ and agent $i$ can not gain from misreporting.

	\textbf{Case 2:} $x_m-b < x_i + b \leq x_m$ (symmetrically $x_m \leq x_i-b < x_m+b$).
	
	Again, the cost of agent $i$ is calculated with respect to $x_i+b$ for both outcomes of the mechanism. This time, it might be that $|(x_i+b)-(x'_m-b)| < |(x_i+b)-(x_m-b)|$ but since $(x'_m-b)-(x_m-b)=(x'_m+b)-(x_m+b)$, it will also hold that $|(x_i+b)-(x'_m+b)| > |(x_i+b)-(x_m+b)|$ and also $|(x_i+b)-(x_m-b)| - |(x_i+b)-(x'_m-b)| = |(x_i+b)-(x'_m+b)| - |(x_i+b)-(x_m+b)|$. Hence, the expected cost of agent $i$ after misreporting is at least as much as it was before.

	\textbf{Case 3:} $x_m < x_i+b \leq x_m+b$ (symmetrically $x_m-b \leq x_i-b < x_m$).
	
	The cost of agent $i$ before misreporting is calculated with respect to $x_i-b$ when the outcome is $x_m-b$ and with respect to $x_i+b$ when the outcome is $x_m+b$. For any misreport $x'_i < x_i+b$, this is still the case (for $x'_m-b$ and $x'_m+b$ respectively) and since  $(x'_m-b)-(x_m-b)=(x'_m+b)-(x_m+b)$, her expected cost is not smaller than before. For any misreport $x'_i > x_i+b$, her cost is calculated with respect to $x_i+b$ for both possible outcomes of the mechanism and for the same reason as in Case 2, her expected cost is at least as much as it was before misreporting.
\end{proof}

\subsection{Social cost}

Next, we will calculate the approximation ratio of the mechanism for the social cost. In order to do that, we will need the following lemma.

\begin{lemma}\label{betweenmlmr}
	Let $\mathbf{x}=\langle x_1,...,x_m,...,x_n\rangle$, where $x_{m}=\mathrm{median}$ $(x_1,...,x_n)$, breaking ties in favor of the smallest index. There exists an optimal location for the social cost in $[x_m-b,x_m+b]$.
\end{lemma}

\begin{proof}
	Assume for contradiction that this is not the case. Then, for any optimal location $y$, it must be that either $y< x_m-b$  or $y>x_m+b$.
	
	Assume first that $y<x_m-b$. Since $x_m$ is the median agent, it holds that for at least $\lceil n/2 \rceil$ agents, $x_i-b \geq x_m-b$, that is $x_m-b$ admits a smaller cost for at least $\lceil n/2 \rceil$ agents when compared to $y$. Let $X_1$ be the set of those agents. On the other hand, for each agent $x_i<x_m$, $x_m-b$ may admit a smaller or larger cost than $y$, depending on her position with respect to $y$. In the worst case, the cost is larger for every one of those agents, which happens when $x_i+b \leq y$ for every agent with $x_i < x_m$. Let $X_2$ be the set of those agents. Now observe that for any two agents $x_a \in X_1$ and $x_b \in X_2$, it holds that $\mathrm{cost}(x_a,y)-\mathrm{cost}(x_a,x_m-b) = \mathrm{cost}(x_b,x_m-b)-\mathrm{cost}(x_b,y)$. Since $|X_1| \geq |X_2|$, it holds that that $SC_{x_m-b}(\mathbf x) \leq SC_{y}(\mathbf x)$. Since it can not be that $SC_{x_m-b}(\mathbf x) < SC_{y}(\mathbf x)$, $x_m-b$ is an optimal location and we get a contradiction. 
	
	Now assume $y> x_m+b$. Let $y=x_m+b$. If the number of agents is odd, then we can use an exactly symmetric argument to prove that $SC_{x_m+b} \leq SC_{y}$. If the number of agents is even, the argument can still be used, since our tie-breaking rule selects agent $x_{n/2}$ as the median. Specifically, $x_m+b$ admits a smaller cost for exactly $n/2$ of the agents (including agent $x_{n/2}$) and in the worst case, $y$ admits a smaller cost for $n/2$ agents as well. If $X_1$ and $X_2$ are the sets of those agents respectively, then again it holds that $\mathrm{cost}(x_a,y)-\mathrm{cost}(x_a,x_m+b) = \mathrm{cost}(x_b,x_m+b)-\mathrm{cost}(x_b,y)$ for $x_a \in X_1$ and $x_b \in X_2$ and we get a contradiction as before.
\end{proof}

We now proceed to proving the approximation ratio of Mechanism~\ref{leftrightmedian}.

\begin{theorem}\label{tiesc}
	Mechanism \ref{leftrightmedian} has an approximation ratio of $1+\frac{b}{c}$ for the social cost.
\end{theorem}

\begin{proof}
	Consider an arbitrary instance $\mathbf{x}=\langle x_1,...,x_n\rangle$ and let $x_m$ be the median agent. By Lemma \ref{betweenmlmr}, there exists an optimal location $y \in [x_m-b,x_m+b]$. Let $\delta=y-(x_m-b)$. For every agent $i$, it holds that $\text{cost}(x_i,x_m-b) \leq \text{cost}(x_i,y)+\delta$. To see this, first observe that $|(x_i-b)-(x_m-b)| \leq |(x_i-b)-y|+\delta$ and  that  $|(x_i+b)-(x_m-b)| \leq |(x_i+b)-y|+\delta$. If the cost of an agent admitted by $y$ and $x_m-b$ is calculated with respect to the same peak, then $\min(|(x_i-b)-(x_m-b)|,|(x_i+b)-(x_m-b)|) \leq \min(|(x_i-b)-y|,|(x_i+b)-y|)+\delta$ and the inequality holds. If the cost is calculated with respect to different peaks for $y$ and $x_m-b$, it must be that $\text{cost}(x_i,x_m-b)=c+|(x_i-b)-(x_m-b)|$ and $\text{cost}(x_i,y)=c+|x_i+b-y|$, because $x_m-b < y$. Since $|(x_i-b)-(x_m-b)| \leq |(x_i+b)-(x_m-b)| \leq |(x_i+b)-y|+\delta$, the inequality holds. Similarily, we can prove that $\text{cost}(x_i,x_m+b) \leq \text{cost}(x_i,y)+(2b-\delta)$ for every agent $i$. Hence, we can upper bound the cost of Mechanism~\ref{leftrightmedian} by
	$\frac{1}{2}\sum_{i=1}^{n}\text{cost}(x_i,x_m-b) +$ $\frac{1}{2}\sum_{i=1}^{n}\text{cost}(x_i,x_m+b)$
	$\leq \frac{1}{2} \sum_{i=1}^{n} \left( \text{cost}(x_i,y)+\delta \right)+\frac{1}{2} \sum_{i=1}^{n} \left( \text{cost}(x_i,y) + 2b-\delta \right)$ $= SC_y(\mathbf{x}) +nb$ = $SC_{\mathrm {opt}}(\mathbf{x}) + nb$.
	The approximation ratio then becomes $1+\frac{nb}{SC_{\mathrm{opt}}(\mathbf{x})}$, which is at most $1+\frac{b}{c}$, since $SC_{\mathrm{opt}}(\mathbf{x})$ is at least $nc$.
	
	For the lower bound, consider the location profile $\mathbf{x}=\langle x_1,...,x_n\rangle$ with $x_1=...=x_{k-1}=x_k-b=x_{k+1}-2b=...=x_n-2b$. Note that the argument works both when $n=2k$ and when $n=2k+1$ because Mechanism \ref{leftrightmedian} selects agent $x_k$ as the median agent in each case. The optimal location is $x_1+c$ whereas Mechanism \ref{leftrightmedian} equiprobably outputs $f_{\mathrm{\ref{leftrightmedian}}}(\mathbf{x})=x_{k}-b$ or $f_{\mathrm{\ref{leftrightmedian}}}(\mathbf{x})=x_{k}+b$. The cost of the optimal location is $SC_{\mathrm{opt}}(\mathbf{x}) = nc + b$ whereas the cost of Mechanism \ref{leftrightmedian} is $SC_{\mathrm{\ref{leftrightmedian}}}(\mathbf{x})=$ $nc + (1/2)(n-1)b$ $+(1/2)(n-1)b$ $=nc+(n-1)b$.
	The approximation ratio then becomes $\frac{nc+(n-1)b}{nc+b}$ $= 1 \frac{b}{c}\cdot\frac{n-2}{n+(b/c)}$. 
	As the number of agents grows to infinity, the approximation ratio of the mechanism on this instance approaches $1+b/c$. This completes the proof.
\end{proof}

\subsection{Maximum cost}

We also consider the maximum cost and prove the approximation ratio of Mechanism \ref{leftrightmedian} as well as a lower bound on the approximation ratio of any truthful-in-expectation mechanism. The results are summarized in Table \ref{resultstable}.

\begin{theorem}\label{tiescmax}
	Mechanism \ref{leftrightmedian} has an approximation ratio of $\max\{1+b/c,2\}$ for the maximum cost.
\end{theorem}

\begin{proof}
	
	Let $\mathbf{x}=\langle x_1,...,x_n\rangle$ be an arbitrary instance and let $x_m$ be the median agent. We will consider two cases, based on the location of $f_{\mathrm{opt}}(\mathbf{x})$ with respect to $x_m-b$ (or symmetrically $x_m+b$). \\
	
	\textbf{Case 1:} $f_{\mathrm{opt}}(\mathbf{x}) < x_m-b$ (or $f_{\mathrm{opt}}(\mathbf{x}) > x_m+b$).
	
	Let $\delta = (x_m-b)-f_{\mathrm{opt}}(\mathbf{x})$. For the same reason as in the proof of Theorem \ref{tiesc}, for every agent $i$, it holds that $\text{cost}(x_i,x_m-b) \leq \text{cost}(x_i,f_{\mathrm{opt}}(\mathbf{x}))+\delta$ and also that $\text{cost}(x_i,x_m+b) \leq \text{cost}(x_i,f_{\mathrm{opt}}(\mathbf{x}))+(2b+\delta)$.
	
	The maximum cost of Mechanism \ref{leftrightmedian} is
	\begin{eqnarray*}
		MC_{\mathrm{\ref{leftrightmedian}}}(\mathbf{x}) &=& \frac{1}{2}\max_{i \in N}\text{cost}(x_i,x_m-b) + \frac{1}{2}\max_{i \in N}\text{cost}(x_i,x_m+b)\\
		&\leq& \frac{1}{2}\left(\max_{i \in N}\text{cost}(x_i,f_{\mathrm{opt}}(\mathbf{x}))+\delta \right) + \frac{1}{2}\left(\max_{i \in N}\text{cost}(x_i,f_{\mathrm{opt}}(\mathbf{x}))+(2b+\delta)\right)\\
		&=& MC_{\mathrm{opt}}(\mathbf{x}) + b + \delta \leq 2MC_{\mathrm{opt}}(\mathbf{x}) +b-c
	\end{eqnarray*}
	since $MC_{\mathrm{opt}}(\mathbf{x}) \geq c+\delta$. The approximation is at most $2+\frac{b-c}{MC_{\mathrm{opt}}(\mathbf{x})}$, which is at most $1+\frac{b}{c}$ if $b\geq c$ (since $MC_{\mathrm{opt}}(\mathbf{x}) \geq c$) and at most $2$ if $b<c$ (since $MC_{\mathrm{opt}}(\mathbf{x}) > 0$).
	
	\textbf{Case 2:} $x_m-c \leq f_{\mathrm{opt}}(\mathbf{x}) \leq x_m+b$.
	
	Now, let $\delta=f_{\mathrm{opt}}(\mathbf{x})-(x_m-b)$. Again, it holds that $\text{cost}(x_i,x_m-b) \leq \text{cost}(x_i,f_{\mathrm{opt}}(\mathbf{x}))+\delta$ and also that $\text{cost}(x_i,x_m+b) \leq \text{cost}(x_i,f_{\mathrm{opt}}(\mathbf{x}))+(2b-\delta)$.
	
	The maximum cost of Mechanism \ref{leftrightmedian} is
	\begin{eqnarray*}
		MC_{\mathrm{\ref{leftrightmedian}}}(\mathbf{x}) &=& \frac{1}{2}\max_{i \in N}\text{cost}(x_i,x_m-b) + \frac{1}{2}\max_{i \in N}\text{cost}(x_i,x_m+b)\\
		&\leq& \frac{1}{2}\left(\max_{i \in N}\text{cost}(x_i,f_{\mathrm{opt}}(\mathbf{x}))+\delta \right) + \frac{1}{2}\left(\max_{i \in N}\text{cost}(x_i,f_{\mathrm{opt}}(\mathbf{x}))+(2b-\delta)\right)\\
		&=& MC_{\mathrm{opt}}(\mathbf{x}) + b 
	\end{eqnarray*}
	and the approximation ratio is at most $1+\frac{b}{c}$ (since $MC_{\mathrm{opt}}(\mathbf{x}) \geq c$).\\
	
	For the matching lower bound, consider an instance $\mathbf{x}=\langle x_1,...,x_n\rangle$ on which $x_1+b < x_2-b$ and $x_i=x_2$ for all $i \notin \{1,2\}$. It is $f_{\mathrm{opt}}(\mathbf{x})= \frac{x_1+x_2}{2}$, i.e. the middle of the interval between $x_1$ and $x_2$, whereas Mechanism \ref{leftrightmedian} selects equiprobably among $x_2-b$ and $x_2+b$. Let $d = f_{\mathrm{opt}}(\mathbf{x}) - (x_1+b)$. Then $MC_{\mathrm{opt}}(\mathbf{x}) = c+d$, whereas $MC_{\mathrm{\ref{leftrightmedian}}}(\mathbf{x}) = c + \frac{1}{2}2d + \frac{1}{2}(2d+2b)= b+c+2d$. The approximation ratio is $1+\frac{b+d}{c+d}$ which is $1+\frac{b}{c}$ as $d$ goes to $0$ and $2$ as $d$ goes to infinity.
	\hfill $\square$
\end{proof}


Next, we provide a lower bound on the approximation ratio of any truthful-in-expectation mechanism.

\begin{theorem}\label{randlowermax}
	Any truthful-in-expectation mechanism has an approximation ratio of at least $\frac{3}{2}$ for the maximum cost.
\end{theorem}

First, we state a couple of lemmas which are in essence very similar to those used in the proof of the single-peaked preferences case in \cite{procaccia2009approximate}. Let $\mathbf{x}=\langle x_1,x_2 \rangle$ be an instance such that $x_1+b<x_2-b$ and let $\lambda = (x_2-b)-(x_1+b)$. Let $f$ be a truthful-in-expectation mechanism and let $\mathcal{D}$ be the distribution that $y=f(\mathbf{x})$ follows on instance $\mathbf{x}$.

\begin{lemma} \label{lemmalambda}
	On instance $\mathbf{x}$, at least one of $\mathbb{E}_{y \sim \mathcal D}\left[\text{cost}(x_1,y)\right]\geq c+\frac{\lambda}{2}$ and \\ $\mathbb{E}_{y \sim \mathcal D}\left[\text{cost}(x_2,y)\right]\geq c+\frac{\lambda}{2}$ holds. 
\end{lemma}

\begin{proof}
	Obviously, $\text{cost}(x_1,y) + \text{cost}(x_2,y) \geq 2c+\lambda$ for any choice of $y$, hence $\mathbb{E}_{y \sim \mathcal D}\left[\sum_{i=1}^{2}\text{cost}(x_i,y)\right] = \sum_{i=1}^{2}\mathbb{E}_{y \sim \mathcal D}[\text{cost}(x_i,y)] \geq 2c+\lambda$. Therefore, it must be that $\mathbb{E}_{y \sim \mathcal D}[\text{cost}(x_i,y)] \geq c+\frac{\lambda}{2}$ for at least one of $i=1$ or $i=2$. 
\end{proof}

\begin{lemma}\label{lemma3lambda}
	Let $f_{\mathrm{opt}}(\mathbf{x})$ be the outcome of the optimal mechanism on instance $\mathbf{x}$. If
	$\mathbb{E}_{y \sim \mathcal  D}[|y-f_{\mathrm{opt}}(\mathbf{x})|] = \Delta$,
	then the  maximum cost of the mechanism on this instance is $\mathbb{E}[MC_f(\mathbf{x})] = c+ \frac{\lambda}{2} + \Delta$.
\end{lemma}

\begin{proof}
	Since $\forall y$, $MC_f(\mathbf{x}) = c + |y-f_{\mathrm{opt}}(\mathbf{x})| + \lambda/2$, it holds
	\begin{align*}
	\mathbb{E}_{y \sim \mathcal  D}[MC_f(\mathbf{x})] = \mathbb{E}_{y \sim \mathcal  D}[c+\frac{\lambda}{2}+|y-f_{\mathrm{opt}}(\mathbf{x})|]=c+\frac{\lambda}{2}+\mathbb{E}_{y \sim \mathcal D}[|y-f_{\mathrm{opt}}(\mathbf{x})|] = c+\frac{\lambda}{2}+\Delta.
	\end{align*}
	\hfill $\square$
\end{proof}

We can now prove the theorem.

\begin{proof}
	Consider an instance with two agents on which $x_1+b < x_2-b$ and $(x_2-b)-(x_1+b)=\lambda$. It holds that $f_{\mathrm{opt}}(\mathbf{x})=\frac{x_1+x_2}{2}$. Assume there is a truthful-in-expectation mechanism $M$ on which $y=f(x_1,x_2)$ follows a distribution $\mathcal D$ on this instance. According to Lemma~\ref{lemmalambda},  at least one of $\mathbb{E}_{y \sim \mathcal  D}[\text{cost}(y,x_1)] \ge c+\lambda/2$ and $\mathbb{E}_{y \sim \mathcal  D}[\text{cost}(y,x_2)] \ge c+\lambda/2$ holds. W.l.o.g., assume the second inequality is true (if the first inequality is true then we can make a symmetric argument with agent $x_1$ deviating).
	
	Next, consider the instance $\mathbf{x'}=\langle x'_1,x'_2 \rangle$ with $x'_1 = x_1$ and $x'_2 = x_2+\lambda$. Let $f_{\mathrm{opt}}(\mathbf{x'}) = (x'_1 + x'_2)/2 = x_2-b$. Let $\mathcal D'$ be the distribution that $y'$ follows on instance $\mathbf{x'}$. By strategyproofness,
	$\mathbb{E}_{y' \sim \mathcal D'} \left[\text{cost}(y',x_2)\right] \ge \mathbb{E}_{y \sim \mathcal  D}[\text{cost}(y,x_2)] \geq c + \lambda/2$, since $x'_2$ could be a deviation of agent 2 on instance $\mathbf{x}$. This implies $\mathbb{E}_{y' \sim \mathcal D'} \left[ |y'-(x_2-b)| \right] \ge \frac{\lambda}{2}$. To see this, assume otherwise for contradiction, we would have
	\begin{eqnarray*}
		\mathbb{E}_{y' \sim \mathcal D'} \left[\text{cost}(y',x_2)\right] &=& \mathbb{E}_{y' \sim \mathcal D'} \left[c+\min\{|y'-(x_2-b)|, |y'-(x_2+b)| \}\right] \\
		&\le& c+\min \left\{ \mathbb{E}_{y' \sim \mathcal D'} \left[ |y'-(x_2-b)| \right], \mathbb{E}_{y' \sim \mathcal D'} \left[ |y'-(x_2+b)| \right] \right\} \\
		&<& c+\frac{\lambda}{2} .
	\end{eqnarray*}
	Equivalently, we have $\mathbb{E}_{y' \sim \mathcal D'} \left[ |y'-f_\mathrm{opt}(\mathbf{x'})| \right] \ge \frac{\lambda}{2}$. Hence, by applying Lemma~\ref{lemma3lambda}, we know that the maximum cost of the mechanism on the second instance is $\mathbb{E}_{y' \sim \mathcal D'}\left[ MC_M(\mathbf{x'}) \right]\ge c+\lambda+\frac{\lambda}{2}=c+\frac{3\lambda}{2}$. Since the optimal mechanism locates the facility on $f_{\mathrm{opt}}(\mathbf{x'})$, its maximum cost is $c+\lambda$. Therefore, the approximation ratio is at least $\frac{c+3\lambda/2}{c+\lambda}$. As $\lambda$ grows to infinity, the approximation ratio approaches $\frac{3}{2}$. To generalize the proof to more than two agents, place every other agent on $f_{\mathrm{opt}}(\mathbf{x})+b$ on instance $\mathbf{x}$. Since $\lambda$ is large enough, the maximum cost is still calculated with respect to $x_2$ and all the arguments still hold.
	\hfill $\square$
\end{proof}

\section{Deterministic Mechanisms}\label{DetMech}

We now turn our attention to deterministic mechanisms. We will start by stating and proving the following lemma, which will be very useful throughout the paper. Variations of the instances used here will appear in several of our proofs.\\
\begin{table}[t]
	\centering
	\caption{Summary of our results. The lower bounds for deterministic mechanisms hold for anonymous and position invariant strategyproof mechanisms. For $\mathbf{(^*)}$, an additional lower bound of $2$ holds under no conditions. Fields indicated by \textbf{(-)} are not proven yet. For the maximum cost, the approximation ratios are actually $\max\{1+2b/c,3\}$ and $\max\{1+b/c,2\}$ respectively. The results for single-peaked preferences are also noted for comparison.}
	\label{resultstable}
	\begin{center}
		\begin{tabular}{ |c  |c  c|  c  c|   }
			\hline
			& \multicolumn{2}{c}{\textbf{Double-peaked}} & \multicolumn{2}{|c|}{\textbf{Single-peaked}}  \\[3pt] 
			& Ratio & Lower & Ratio & Lower \\[2pt] \hline
			\textbf{Social cost}  & &  &  &   \\[2pt] 
			Deterministic & $\Theta(n)$ & $1+ \frac{b}{c}$  & $1$ & $1$ \\[2pt] 
			Randomized & $1+\frac{b}{c}$ & -  & $1$ & $1$ \\[2pt]
			\hline
			\textbf{Maximum cost} &  &  &  &    \\[3pt] 
			Deterministic$^*$ & $1+\frac{2b}{c}$ & $1+ \frac{2b}{c}$  & $2$ & $2$  \\[2pt]
			Randomized & $1+\frac{b}{c}$ & $3/2$ & $3/2$ & $3/2$ \\[2pt]
			\hline
		\end{tabular}
	\end{center}
\end{table}

When $n=2$, we call the following instance\footnote{We note that this is actually a family of different instances, for each possible value of $x_1$ on the real line. Any one of those instances can be used for our purposes and hence we use the term loosely.} the {\em primary instance}.

{\bf Primary instance:} $\mathbf x=\langle x_1,x_2 \rangle$ with $x_1+2b+\epsilon = x_2-b$, where $\epsilon$ is an arbitrarily small positive quantity.

\begin{lemma}\label{Pinstance}
On the primary instance, there is no anonymous, position invariant and strategyproof mechanism such that  $f(\mathbf x) \in [x_1+b,x_2-b]$.
\end{lemma}

\begin{figure}
\centering
\includegraphics[scale=0.4]{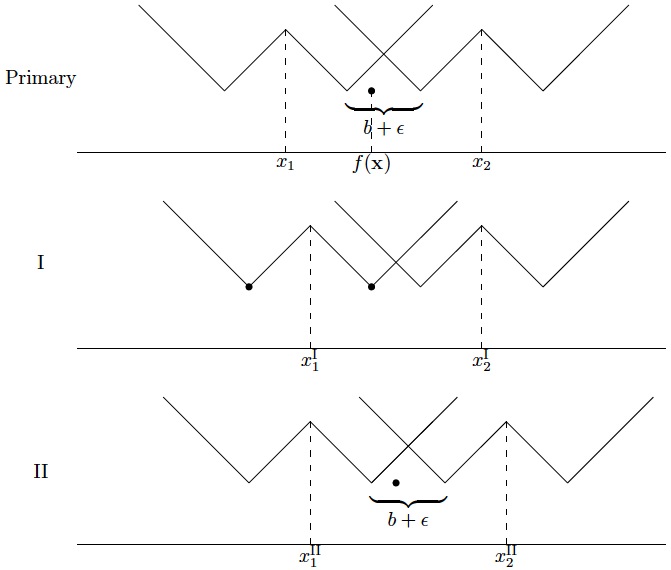}\label{groupstrategyproofcase1}
\caption{Case 1 of Lemma~\ref{Pinstance}\label{Pinstance1}. On instance $\mathrm{I}$, agent $x_2^\mathrm{I}$ can report $x_2^{\mathrm{II}}$ and decrease her cost. Instances $\mathrm{II}$ and the Primary instance are position equivalent.}
\end{figure}

\begin{figure}
	\begin{subfigure}{.5\textwidth}
\centering
\includegraphics[scale=0.4]{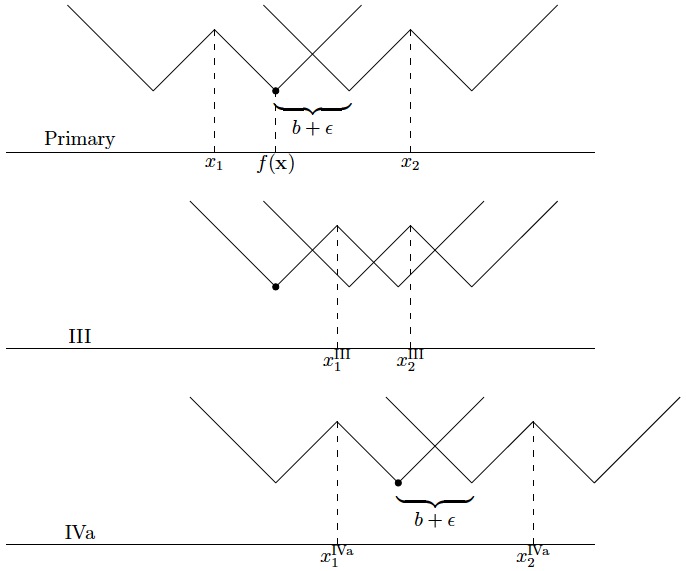}\label{Pinstance2}
\end{subfigure}
	\begin{subfigure}{.5\textwidth}
		\centering
		\includegraphics[scale=0.4]{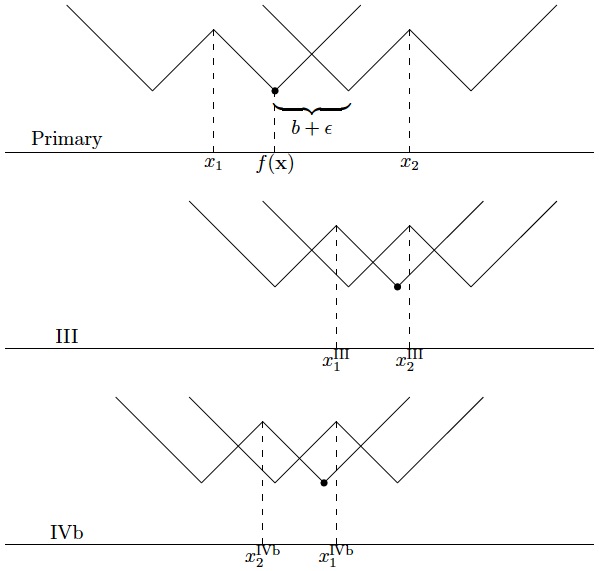}\label{Pinstance2}
	\end{subfigure}
\caption{Case 2(a) and 2(b) of Lemma~\ref{Pinstance}\label{Pinstance2}. If $f(x_1^{\mathrm{III}},x_2^{\mathrm{III}})=x_1^{\mathrm{III}}-b$ on instance $\mathrm{III}$, agent $x_2^{\mathrm{III}}$ can misreport $x_2^{\mathrm{IVa}}$ and decrease her cost, since instances $\mathrm{IVa}$ and the Primary instance are position equivalent. If $f(x_1^{\mathrm{III}},x_2^{\mathrm{III}})=x_1^{\mathrm{III}}+b$ on instance $\mathrm{III}$, then on the Primary instance, agent $x_2$ can misreport $x_2^{\mathrm{IVb}}$ and decrease her cost. By anonymity, instances $\mathrm{III}$ and $\mathrm{IVb}$ are position equivalent.}
\end{figure}

\begin{proof}
For contradiction, suppose there exists an anonymous,  position invariant   strategyproof mechanism $M$ that ouputs $f(\mathbf x)\in [x_1+b,x_2-b]$. Let's denote $\delta_1 = f(\mathbf{x}) - (x_1+b)$, $\delta_2=(x_2-b)-f(\mathbf{x})$. Throughout the proof we start with the primary instance and construct some other instances to prove the lemma. There are 2 cases to be considered.\\

\textbf{Case 1:} $0<\delta_1 \le \frac{1}{2} (b+\epsilon)$, i.e., $x_1+b<f(\mathbf{x}) \le \frac{x_1+x_2}{2}$ (symmetrically, for the case $0<\delta_2 \le \frac{1}{2} (b+\epsilon)$). See Figure~\ref{Pinstance1}.
\begin{itemize}
\item  Instance  $\mathrm{I}$ :  $\mathbf{x^{\mathrm{I}}}=\langle x^{\mathrm{I}}_1,x^{\mathrm{I}}_2 \rangle$, where $x^{\mathrm{I}}_1=x_1+\delta_1,  x^{\mathrm{I}}_2=x_2$.
\item  Instance $\mathrm{II}$ :  $\mathbf{x^{\mathrm{II}}}=\langle x^{\mathrm{II}}_1,x^{\mathrm{II}}_2 \rangle$, where $x^{\mathrm{II}}_1=x_1+\delta_1,  x^{\mathrm{II}}_2=x_2+\delta_1$.
\end{itemize}

First, we know that on instance  $\mathrm{I}$,  strategyproofness requires that $f(\mathbf {x^{\mathrm{I}}})=x_1^{\mathrm{I}}-b$ or $x_1^{\mathrm{I}}+b$, otherwise agent $x^{\mathrm{I}}_1$ could misreport $x_1$ and move the facility to position $x_1^{\mathrm{I}}+b$, minimizing her cost. Second, since instance $\mathrm{II}$ and the primary instance are position equivalent ($\mathrm{II}$ is a shifted version of the primary instance to the right by $\delta_1$), by position invariance we know $f(\mathbf {x^{\mathrm{II}}})=x_1^{\mathrm{II}}+b+\delta_1=x_1^{\mathrm{I}}+b+\delta_1$. Now let's consider agent $x_2^{\mathrm{I}}$ in instance $\mathrm{I}$. If agent $x_2^{\mathrm{I}}$ misreports $x_2^{\mathrm{II}}$, she can move the facility from $x_1^{\mathrm{I}}-b$ or $x_1^{\mathrm{I}}+b$ to $x_1^{\mathrm{I}}+b+\delta_1 \le x^{\mathrm{I}}_2-b$; in either case the facility is closer to her left peak. Therefore $x_2^{\mathrm{I}}$ can manipulate the mechanism, which violates strategyproofness. Hence $f(\mathbf{x}) \not\in (x_1+b,\frac{x_1+x_2}{2}]$.\\

\textbf{Case 2:} $\delta_1=0$, i.e., $f(\mathbf{x})=x_1+b$ (symmetrically, for the case $\delta_2=0$). See Figure~\ref{Pinstance2}.
\begin{itemize}
\item  Instance  $\mathrm{III}$ :  $\mathbf{x^{\mathrm{III}}}=\langle x^{\mathrm{III}}_1,x^{\mathrm{III}}_2 \rangle$, where $x^{\mathrm{III}}_1=x_1+2b,  x^{\mathrm{III}}_2=x_2$.
\end{itemize}

First, we know that on instance  $\mathrm{III}$,  strategyproofness requires that $f(\mathbf{x^{\mathrm{III}}})=x_1^{\mathrm{III}}-b$ or $x_1^{\mathrm{III}}+b$, otherwise agent $x^{\mathrm{III}}_1$ could misreport $x_1$ and move the facility to position $x_1^{\mathrm{III}}-b$, minimizing her cost. Now, there are two subcases.

\begin{enumerate}[(a)]
\item $f(\mathbf{x^{\mathrm{III}}})=x_{1}^{\mathrm{III}}-b$. See left hand side of Figure~\ref{Pinstance2}, where we have instance $\mathrm{IVa}$: $x_1^{\mathrm{IVa}}=x_1+2b, x_2^{\mathrm{IVa}}=x_2+2b$.


Obviously, instance $\mathrm{IVa}$ is position equivalent to the primary instance, so $f(x_1^{\mathrm{IVa}},x_2^{\mathrm{IVa}})=x_1^{\mathrm{IVa}}+b$. Note that the cost of agent $x_2^{\mathrm{III}}$ is $b+c+\epsilon$. If agent $x_2^{\mathrm{III}}$ misreports $x_2^{\mathrm{IVa}}$, then her cost becomes $b+c-\epsilon$, which is smaller than when reporting truthfully. So $f(\mathbf{x^{\mathrm{III}}})\not=x_1^{\mathrm{III}}-b$ if Mechanism $M$ is strategyproof.\\

\item $f(\mathbf{x^{\mathrm{III}}})=x_1^{\mathrm{III}}+b$. See right hand side of  Figure~\ref{Pinstance2}, where we have instance $\mathrm{IVb}$: $x_1^{\mathrm{IVb}}=x_1^{\mathrm{III}}, x_2^{\mathrm{IVb}}=x_1^{\mathrm{III}}-(b+\epsilon)$.


Obviously, instance $\mathrm{IVb}$ is position equivalent to instance $\mathrm{III}$, so $f(x_1^{\mathrm{IVb}},x_2^{\mathrm{IVb}})=x_2^{\mathrm{IVb}}+b$ as implied by position invariance and anonymity. Note that on instance $\mathrm{III}$, the cost of agent $x_2^{\mathrm{III}}$ is $b+c-\epsilon$. If agent $x_2^{\mathrm{III}}$ misreports $x_2^{\mathrm{IVb}}$, then her cost becomes $c+2\epsilon$, which is smaller than when reporting truthfully (since $\epsilon$ is arbitrarily small). So $f(\mathbf{x^{\mathrm{III}}})\not=x_1^{\mathrm{III}}+b$ if Mechanism $M$ is strategyproof.
\end{enumerate}

Hence, in order for mechanism $M$ to be strategyproof, it must be $f(\mathbf{x}) \not= x_1+b$.

In all, there is no anonymous, position invariant and strategyproof mechanism such that  $f(\mathbf x) \in [x_1+b,x_2-b]$.
\hfill $\square$
\end{proof}

\subsection{Group strategyproofness}

As we mentioned in the introduction, under the reasonable conditions of position invariance and anonymity, there is no group strategyproof mechanism for the problem. We will prove this claim by using Lemma \ref{Pinstance} and the following lemma.

When $n=2k+1, k\in \mathbb Z^+$, let $P$ be the instance obtained by locating $k+1$ agents on $x_1$ and $k$ agents on $x_2$ on the primary instance. Similarly, let $S$ be the instance obtained by locating $k$ agents on $x_1$ and $k+1$ agents on $x_2$ on the primary instance. Formally, let $\mathbf x^P=\langle x_1^P, \dots,x_n^P \rangle$, where $x_1^P=\dots=x_{k+1}^P$, $x_{k+2}^P=\dots=x_{n}^P$, and $(x_{n}^P-b)-(x_{1}^P+b)=b+\epsilon$ and $\mathbf x^S=\langle x_1^S, \dots,x_n^S \rangle$, where $x_1^S=\dots=x_{k}^S$, $x_{k+1}^S=\dots=x_{n}^S$, and $(x_{n}^S-b)-(x_{1}^S+b)=b+\epsilon$, where $\epsilon$ is the same quantity as in the primary instance.


\begin{lemma}\label{GSPlemma}
When $n=2k+1$, any  position invariant and group strategyproof mechanism that outputs $f(\mathbf{x^P})=x_{1}^{\mathrm P}+b$ on instance $\mathrm P$, must output $f(\mathbf{x})=x_1+b$ on any instance $\mathbf{x}=\langle x_1,...,x_n \rangle$, where $x_1=...=x_{k+1}$, $x_{k+2}=...=x_{n}$ and $(x_{n}-b)-(x_{1}+b)=2b$. Similarly, any  position invariant  and group strategyproof mechanism that outputs $f(\mathbf{x^S})=x_{n}^{\mathrm S}-b$ on instance $\mathrm S$, must output $f(\mathbf{x})=x_n-b$ on any instance  $\mathbf{x}=\langle x_1,...,x_n \rangle$, where $x_1=...=x_{k}$, $x_{k+1}=...=x_{n}$ and $(x_{n}-b)-(x_{1}+b)=2b$.
\end{lemma}


\begin{proof}
We prove the first part of the lemma. The proof of the second part is symmetric. Note that the difference between instances $\mathbf x^{\mathrm P}$ and $\mathbf x$ is that the distance between the two groups of agents is $b+\epsilon$ in $\mathbf x^{\mathrm P}$ while it is $2b$ in $\mathbf x$. First, we argue that $f(\mathbf x)\in [x_1+b,x_n-b]$. Indeed, if that was not the case, if $f(\mathbf x)<x_1+b$, by the onto condition implied by position invariance, all agents could jointly misreport some different positions and move the facility to $x_1+b$. This point admits a smaller cost for all agents; specifically the cost of agents $x_1,\dots,x_{k+1}$ is minimized while the cost of agents $x_{k+2},\dots,x_{n}$ is reduced and group strategyproofness is violated. Using a symmetric argument, we conclude that it can't be $f(\mathbf x)>x_n-b$ either.

Secondly, we argue that $f(\mathbf x) \not\in (x_1+b,x_n-b]$. Indeed, assume that was not the case. Then agents $x^{\mathrm P}_{k+2},\dots,x^{\mathrm P}_{n}$ on instance $\mathbf x^{\mathrm P}$ could jointly misreport $x_{k+2},\dots,x_{n}$ and move the facility from $f(\mathbf{x^P})=x_{1}^{\mathrm P}+b$ to $f(\mathbf x)$. Since by assumption $f(\mathbf x) \in (x_1+b,x_n-b]$, the cost of each deviating agent is smaller than her cost before deviating. Group strategyproofness is then violated and hence it must be that $f(\mathbf{x})=x_1+b$. By position invariance, it must be that $f(\mathbf{\bar{x}})=\bar{x}_1+b$ on any instance $\mathbf{\bar{x}}$ which is position equivalent to instance $\mathbf{x}$.
\hfill $\square$
\end{proof}

\begin{theorem}\label{groupstrategyproof}
There is no group strategyproof mechanism that is anonymous and position invariant. \label{gsp}
\end{theorem}

\begin{proof}
When $n=2$, on the primary instance, according to Lemma~\ref{Pinstance}, there is no anonymous, position invariant and strategyproof mechanism such that  $f(\mathbf x) \in [x_1+b,x_2-b]$. In addition, if the facility was placed on some point $f(\mathbf{x}) < x_1+b$ (the argument for $f(\mathbf{x})>x_2-b$ is symmetric), for any mechanism that satisfies position invariance which implies onto, agents $1$ and $2$ could jointly misreport some positions $\hat{x}_1$ and $\hat{x}_2$ such that $f(\hat{x}_1,\hat{x}_2)=x_1+b$. Obviously, this deviation admits the minimum possible cost for agent $1$ and a reduced cost for agent $2$, violating group strategyproofness.

The proof can easily be extended to the case when $n$ is even. On the primary instance, simply place $\frac{n}{2}$  agents on $x_1$ and $\frac{n}{2}$  agents on $x_2$. By considering deviations of coalitions of agents coinciding on $x_1$ or $x_2$ instead of deviations of agents $x_1$ and $x_2$ respectively, all the arguments still hold. However, additional care must be taken when $n$ is odd.

When $n=2k+1, k\in \mathbb{Z^+}$, we denote by $P_{\mathrm{J}}$ the instance after placing $k+1$ agents on $x_{1}^{\mathrm{J}}$ and $k$ agents on $x_{2}^{\mathrm{J}}$ on instance $\mathrm{J}$, where $\mathrm{J}$ is either instance $\mathrm{I}$, $\mathrm{II}$, $\mathrm{III}$, $\mathrm{IVa}$ or $\mathrm{IVb}$ of the proof of Lemma \ref{Pinstance}. Similarly, let $S_{\mathrm{J}}$ be the instance after placing $k$ agents on $x_{1}^{\mathrm{J}}$ and $k+1$ agents on $x_{2}^{\mathrm{J}}$ on  instance $\mathrm{J}$. Finally, let ${X_{1}^{\mathrm{K}}}$ be the group of agents $i$ for which $x_{i}^{\mathrm K}=x_{1}^{\mathrm{K}}$ on instance $K$, where $K$ is either $P,P_{\mathrm{I}},P_{\mathrm{II}},P_{\mathrm{III}},P_{\mathrm{IVa}},P_{\mathrm{IVb}}, S,S_{\mathrm{I}},S_{\mathrm{II}},S_{\mathrm{III}}, S_{\mathrm{IVa}}$ or $S_{\mathrm{IVb}}$, and let ${X_{2}^{\mathrm{K}}}$ be the group of agents $i$ for which $x_{i}^{\mathrm K}=x_{n}^{\mathrm{K}}$ on instance $\mathrm{K}$.

For contradiction, assume that there exists an anonymous, position invariant, and group strategyproof mechanism. On instance $P$, by group strategyproofness, it must be $f(\mathbf{x^{P}}) \in [x_1^{P}+b,x_n^{P}-b]$ and by the same arguments used in case 1 of Lemma \ref{Pinstance} (with instances $P_{\mathrm{I}},P_{\mathrm{II}}$ instead of $\mathrm{I},\mathrm{II}$ and with $X_{1}^{P}$ and $X_{2}^{P}$ instead of $x_1$ and $x_2$)\footnote{Here we assume for convenience that any mechanism outputs the same location in $[x_1^{P}+c,x_n^{P}-c]$ on the primary instance and instance $P$. This is without loss of generality because the argument for any output in $[x_1^{P}+b,x_n^{P}-b]$ is exactly the same.}, it must be $f(\mathbf{x^P}) \notin (x_1^{P}+b,x_n^{P}-b)$. Hence, it must be that either $f(\mathbf{x^P})=x_1^{P}+b$ or $f(\mathbf{x^P})=x_n^{P}-b$. Assume w.l.o.g. that $f(\mathbf{x^P})=x_1^{P}+b$; the other case can be handled symmetrically.

Following the arguments of case 2 of Lemma \ref{Pinstance} (using instance $P_{\mathrm{III}}$ instead of $\mathrm{III}$), group strategyproofness and position invariance imply that $f(\mathbf{x^{P_{\mathrm{III}}}}) = x_{1}^{P_{\mathrm{III}}}-b$ or $f(\mathbf{x^{P_{\mathrm{III}}}}) = x_{1}^{P_{\mathrm{III}}}+b$ on instance $P_{\mathrm{III}}$. By the arguments of subcase (a) (using instance $P_{\mathrm{IVa}}$ instead of $\mathrm{IVa}$), it can't be that $f(\mathbf{x^{P_{\mathrm{III}}}}) = x_{1}^{P_{\mathrm{III}}}-b$, so it must be that $f(\mathbf{x^{P_{\mathrm{III}}}}) = x_{1}^{P_{\mathrm{III}}}+b$. However, we can not simply apply the argument used in subcase (b) to get a contradiction, because since there is a different number of agents on $x_{1}^P$ and $x_{2}^P$, instances $P$ and $P_{\mathrm{IVa}}$ are no longer position equivalent. 

Next, consider instance $S$ and observe that $f(\mathbf{x^S}) \notin (x_{1}^{S}+b,x_{n}^{S}-b)$ by the same arguments as above and $f(\mathbf{x^S}) \in [x_{1}^{S}+b,x_{n}^{S}-b]$ by group strategyproofness and position invariance. Hence it is either $f(\mathbf{x^S})= x_{1}^{S}+b$ or $f(\mathbf{x^S})= x_{n}^{S}-b$. Assume first that $f(\mathbf{x^S})=x_{1}^{S}+b$. By the same arguments as above (using instances $S_{\mathrm{III}}$ and $S_{\mathrm{IVa}}$), on instance $S_{\mathrm{III}}$, it must be that $f(\mathbf{x^{S_{\mathrm{III}}}}) = x_{1}^{S_{\mathrm{III}}}+b$. Now, observe that if $X_{2}^{P_{\mathrm{III}}}$ misreport $\bar{x}_{i}^{P_{\mathrm{III}}}=x_{i}^{P_{\mathrm{III}}}-2b-2\epsilon$, then we get instance $P_{\mathrm{IVb}}$ which is position equivalent to instance $S_{\mathrm{III}}$ and hence it must be that $f(\mathbf{x^{P_{\mathrm{IVb}}}})=x_n^{P_{\mathrm{IVb}}}+b$. The cost of $X_{2}^{P_{\mathrm{III}}}$ before misreporting was $b+c-\epsilon$ while it becomes $c+2\epsilon$ after misreporting. This violates strategyproofness, which means that on instance $S$, it must be $f(\mathbf{x^S})=x_{n}^{S}-b$.

Let's denote instance $\mathrm{T}$ by $\mathbf{x^T}=\langle x_1^{\mathrm{T}},...,x_n^{\mathrm{T}} \rangle$, where  $x_1^{\mathrm{T}}=...=x_{k}^{\mathrm{T}}=x_{k+1}^{\mathrm{T}}-2b=x_{k+2}^{\mathrm{T}}-4b=...=x_{n}^{\mathrm{T}}-4b$. Let $X_{1}^{\mathrm{T}}$ be the set of agents $i$ for which $x_i^{\mathrm{T}}=x_1^{\mathrm{T}}$, and let $X_{2}^{T}$ be the set of agents $j$ for which $x_j^{\mathrm{T}}=x_n^{\mathrm{T}}$, and let $x_{t}$ be agent $x_{k+1}^{\mathrm{T}}$. On instance $\mathrm{T}$, it must be either $f(\mathbf{x}^{T})=x_{1}^{\mathrm{T}}+b$ or $f(\mathbf{x}^{T})=x_{n}^{\mathrm{T}}-b$, otherwise agent $x_{t}$ could misreport $x_{t}'=x_{1}^{\mathrm{T}}$ and then by Lemma \ref{GSPlemma} and the fact that on instance $P$ it is $f(\mathbf{x^P})=x_1^{P}+b$, it should be $f(x_1^{\mathrm{T}},...,x_{t}',...,x_n^{\mathrm{T}})=x_t-b$, which admits a cost of $c$ for agent $x_t$. Similarily, by Lemma \ref{GSPlemma} and the fact that on instance $S$, it is $f(\mathbf{x^S})=x_n^{S}-b$, if agent $x_t$ misreports $x_{t}''= x_{n}^{\mathrm{T}}$, then it should be $f(x_1^{\mathrm{T}},...,x_{t}'',...,x_n^{\mathrm{T}})=x_t+b$. If $f(\mathbf{x^{T}})=x_{n}^{\mathrm{T}}-b$, agent $x_t$ could form a coalition with agents $X_{1}^{\mathrm{T}}$ and by misreporting $x_{t}'$, move the facility to $x_{t}-b$, a choice that would admit the same cost for her, but a strictly smaller cost for every other member of the coalition. If $f(\mathbf{x^{T}})=x_{1}^{\mathrm{T}}+b$ then agent $x_t$ could form a coalition with agents $X_{2}^{\mathrm{T}}$ and by misreporting $x_{t}''$, move the facility to $x_t+b$, a choice that would admit the same cost for her, but a strictly smaller cost for every other member of the coalition. In each case, there is a coalition of agents that can benefit from misreporting and group strategyproofness is violated. This completes the proof.
\hfill $\square$
\end{proof}


\subsection{Strategyproofness}

We now turn our attention to (simple) strategyproofness. It is well known \cite{moulin1980strategy} that when agents have single-peaked preferences, outputing the location of the $k$th agent ($k$th order statistic), results in a strategyproof mechanism. This is not the case however, for double-peaked preferences and any choice of $k$.

\begin{lemma}\label{nokorder}
Given any instance $\mathbf{x}=\langle x_1,...,x_n \rangle$,  any mechanism that outputs $f(\mathbf x)=x_i-b$, for $i=2,...,n$ or any mechanism that outputs $f(\mathbf x)=x_i+b$, for $i=1,...,n-1$, is not strategyproof.
\end{lemma}

\begin{proof}
We prove the case when $f(\mathbf x)=x_1+b$. The arguments for the other cases are similar. Consider any instance where $x_2=x_1+b$ and $x_i-b > x_2+b$ for $i=3,...,n$. Since $f(\mathbf{x})=f(x_1,x_2,x_3,...,x_n)=x_1+b$, the cost of agent 2 is $\mathrm{cost}(f(\mathbf{x}),x_2)=b+c$. If agent 2 misreports $x_2'=x_1-b$, the outcome will be $f(\mathbf{x'})=f(x_1,x_2',x_3,...,x_n)=x_2'+c=x_1$, and $\mathrm{cost}(f(\mathbf{x'}),x_2)=c<b+c=\mathrm{cost}(f(\mathbf{x}),x_2)$. Agent $2$ has an incentive to misreport, therefore the mechanism is not strategyproof.
\hfill $\square$
\end{proof}

This only leaves two potential choices among $k$th order statistics, either $f(\mathbf x)=x_{1}-b$ or $f(\mathbf x)=x_{n}+b$. Consider the following mechanism.

\begin{mechanism}{M2}\label{mech}
Given any instance $\mathbf{x}=\langle x_1,...,x_n \rangle$, locate the facility always on the left peak of agent 1,  i.e. $f(\mathbf x)=x_{1}-b$ or always on the right peak of agent $n$, i.e. $f(\mathbf{x})=x_n+b$.
\end{mechanism}

From now on, we will assume that \ref{mech} locates the facility on $x_1-b$ on any instance $\mathbf{x}$. The analysis for the other case is similar.

\begin{theorem}
Mechanism \ref{mech} is strategyproof.
\end{theorem}

\begin{proof}
Obviously, agent $1$ has no incentive to misreport, since her cost is already minimum. For any other agent $i, i=2,\dots,n$, the cost is $\mathrm{cost}(f(\mathbf x),x_{i}) = c+x_i-x_1$. For any misreport $x'_{i} \geq x_1$, the facility is still located on  $x_1-b$ and every agent's cost is the same as before. For some misreport $x'_{i} < x_{1}$, the facility moves to $f(x_1,...,x_i',...,x_n)=x'_{i}-b<x_1-b$, i.e. further away from any of agent $i$'s peaks and hence this choice admits a larger cost for her.
\hfill $\square$
\end{proof}

In the following, we prove that for the case of two agents, Mechanism \ref{mech} is actually the only strategyproof mechanism that satisfies anonymity and position invariance. We start with the following lemma.

\begin{lemma} \label{leftpeak}
For any instance  $\mathbf x= \langle x_1,x_2 \rangle$, where  $x_1+b<x_2-b$, if an anonymous, position invariant and  strategyproof mechanism outputs $f(x_1,x_2)=x_1-b$, then it must output $f(x'_1,x'_2)=x_1'-b$ for any instance $\mathbf{x'}=\langle x_1',x_2'\rangle $ with $x_1' \le x_2'$.
\end{lemma}

\begin{proof}
Let $(x_2-b) - (x_1+b) = \gamma$. For any other instance $\mathbf{x'}=\langle x'_1,x'_2\rangle $ with $x_1' \le x_2'$, we assume $x'_1=x_1$ without loss of generality due to position invariance. We will consider potential deviations of agent $2$ of instance $\mathbf x$. Note that position invariance and anonymity suggest that we only need to consider deviations $x_2'\geq x_1$ because any instance on which $x_2'< x_1$ is position equivalent to some instance $\langle \hat{x}_1,\hat{x}_2\rangle$ with $\hat{x}_1=x_1$ and $\hat{x}_2 > x_1$.

On instance $\mathbf{x}$, the cost of agent $2$ is $2b+c+\gamma$, so for any deviation of $x_2$, her cost must be at least $2b+c+\gamma$, as required by strategyproofness. This implies that if agent $2$ misreports $x'_2$, then on the resulting instance it must be either $f(x_1,x'_2) \in (-\infty,x_1-b]$ or $f(x_1,x'_2)\in [x_2+3b+\gamma,+\infty)$.

First, assume $f(x_1,x_2')\in [x_2+3b+\gamma,\infty)$. Let instance $\mathrm{I}$ be $\mathbf{x^I}=\langle x^\mathrm{I}_1,x^\mathrm{I}_2\rangle $, where $x^\mathrm{I}_1=x_1, x^\mathrm{I}_2=f(x_1,x_2')+b$. On instance $\mathrm{I}$ it must be either $f(x^\mathrm{I}_1,x^\mathrm{I}_2)=f(x_1,x_2')$ or $f(x^\mathrm{I}_1,x^\mathrm{I}_2)=f(x_1,x_2')+2b$, otherwise agent 2 can deviate from $x^\mathrm{I}_2$ to $x_2'$ and move the facility to $x_2^{\mathrm{I}}-b$, minimizing her cost and violating strategyproofness. However then, on instance $\mathrm{I}$, agent $1$ can misreport $\bar{x}_1^{\mathrm{I}}=x^\mathrm{I}_2-2b-\gamma$ and move the facility to $x^{\mathrm{I}}_2-3b-\gamma$, reducing her cost and violating strategyproofness. This follows from the fact that the resulting instance $\langle \bar{x}_1^\mathrm{I},x_2^\mathrm{I}\rangle$ and instance $\mathbf{x}$ are position equivalent and that $f(x_1,x_2)=x_1-b$ on instance $\mathbf{x}$. Hence, it can not be that $f(x_1,x_2') \in [x_2+3b+\gamma,\infty)$ on instance $\mathbf{x'}$.

Second, assume $f(x_1,x_2')\in (-\infty,x_1-b)$. Then, since $x_2' > x_1$, agent 2 can deviate from $x_2'$ to $x_2$ and move the facility to $x_1-b$, i.e. closer to her actual position. Hence, it can't be that $f(x_1,x_2') \in (-\infty,x_1-b)$ on instance $\mathbf{x'}$ either.

In conclusion, it must be $f(x_1,x_2')=f(x_1',x_2')=x_1'-b$ for any instance $\mathbf{x'}=(x_1',x_2')$ with $x_1' \le x_2'$.
\end{proof}


\begin{theorem}\label{2unique}
When $n=2$, the only strategyproof mechanism that satisfies position invariance and anonymity is Mechanism \ref{mech}.
\end{theorem}

\begin{proof}\label{A2unique}
For contradiction, suppose there exists an anonymous, position invariant strategyproof mechanism $M$ which is different from Mechanism~\ref{mech} and consider the primary instance used in Lemma \ref{Pinstance}.

We first argue that it must be that $f(x_1,x_2) \in [x_1+b,x_2-b]$. Assume on the contrary that $f(x_1,x_2) < x_1+b$ (the argument for $f(x_1,x_2)>x_2-b$ is symmetric). Then, consider the instance $\langle x'_1,x'_2 \rangle$, where $x_2'=x_2$ and $x_1'=f(x_1,x_2)-b$. On this instance, it must be that $f(x_1',x_2') = x_1'-b$ or $f(x_1',x_2') = x_1'+b$, otherwise agent 1 can deviate from $x_1'$ to $x_1$ and move the facility to $x_1'+b$, minimizing her cost. In addition, if $f(x_1',x_2') = x_1'-b$, then according to Lemma~\ref{leftpeak}, the unique strategyproof mechanism that satisfies position invariance and anonymity is Mechanism \ref{mech}, and thus our assumption is violated. So let's assume $f(x_1',x_2')=x_1'+b$. Then, on the primary instance, agent $2$ could report $\hat{x}_2 =x_2+(x_1+b-f(x_1,x_2))$ and by position invariance (since instances $\langle x_1',x_2'\rangle$ and $\langle x_1,\hat{x}_2\rangle$ are position equivalent), it should be $f(x_1,\hat{x}_2)=x_1+b$. This would give agent $2$ an incentive to misreport, violating strategyproofness. Hence, on the primary instance it must be that $f(x_1,x_2) \in [x_1+b,x_2-b]$.

However, according to Lemma~\ref{Pinstance}, it is impossible for a strategyproof mechanism to output $f(x_1,x_2) \in [x_1+b,x_2-b]$ on the primary instance. In conclusion, no mechanism other than Mechanism \ref{mech} is strategyproof, anonymous and position invariant.
\hfill $\square$
\end{proof}

\subsection{Social Cost}

We have seen that Mechanism \ref{mech} is strategyproof, but how well does it perform with respect to our goals, namely minimizing the social cost or the maximum cost? In other words, what is the approximation ratio that Mechanism \ref{mech} achieves against the optimal choice? First, we observe that the optimal mechanism, which minimizes the social cost, is not strategyproof.

\begin{theorem}
The optimal mechanism with respect to the social cost, $f_{\mathrm{opt}}(\mathbf{x}) = \arg \min\limits_y \sum\limits_{i=1}^{n} \mathrm{cost}(y,x_i)$, is not strategyproof.
\end{theorem}

\begin{proof}
Consider an instance $\mathbf x=\langle x_1,x_2,x_3\rangle $, such that  $x_2+b<x_3-b$, $x_2-b < x_1 + b < x_2$ and  $(x_1+b)-(x_2-b) = \epsilon$, where $\epsilon$ is an arbitrarily small positive quantity. On this instance, the optimal facility location is $x_2+b$ and the cost of agent $x_1$ is $c+2b-\epsilon$. Suppose now that agent $x_1$ reports $x'_1 < x_2-2b$. Moreover, suppose that when there are two locations $y_1$ and $y_2$, with $y_1 < y_2$ that admit the same social cost, the optimal mechanism outputs $y_1$. If the mechanism outputs $y_2$ instead, we can use a symmetric argument on the instance $\mathbf{x'}=\langle x'_1,x'_2,x'_3 \rangle =\langle x'_1,x_2,x_2+2b-\epsilon\rangle$ with agent $3$ misreporting $x_3$.\footnote{In fact, even if the optimal mechanism outputs a distribution over points that all admit the minimum social cost, the argument still works.} By this tie-breaking rule, on instance $\mathbf{x}=\langle x_1',x_2,x_3\rangle$, the location of the facility is $x_2-b$ and the cost of agent $x_1$ is $c+\epsilon$, i.e. smaller than before. Hence, the optimal mechanism is not strategyproof. To extend this to an arbitrary number of agents, let $x_j=x_2$ for every other agent $x_j$.
\hfill $\square$
\end{proof}

Unfortunately, when considering the social cost, in the extremal case, the approximation ratio of Mechanism \ref{mech} is dependent on the number of agents. The approximation ratio is given by the following theorem.

\begin{theorem}\label{deter-ratio}
For $n > 3$ agents, Mechanism \ref{mech} achieves an approximation ratio of $\Theta(n)$ for the social cost.
\end{theorem}

\begin{figure}
\centering
\includegraphics[scale=0.5]{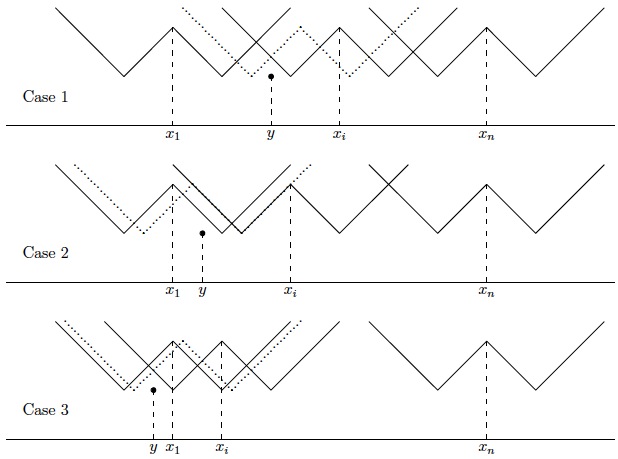}
\caption{Cases of Theorem~\ref{deter-ratio}.\label{detercase} The dotted lines indicate the position of agent $x_i$ before moving her from the left side to the right side of $y$. The cost of every agent with respect to $y$ is the same before and after moving her to the right. The figure depicts only three agents (including $x_1$ and $x_n$).}
\end{figure}

\begin{proof}

Consider any instance $\mathbf{x'}=\langle x'_1,...,x'_n\rangle$ and let $y=f_{\mathrm{opt}}(\mathbf{x'})$. It also holds that $f_{\mathrm{\ref{mech}}}(\mathbf{x'})=x'_1-b$. Denote the social costs of the optimal mechanism and Mechanism \ref{mech} on instance $\mathbf{x'}$ by $SC_{\mathrm{opt}}(\mathbf{x'})$ and $SC_{\mathrm{\ref{mech}}}(\mathbf{x'})$, respectively. Let $\mathbf{x}$ be the instance obtained by instance $\mathbf{x'}$ as follows. For every agent $i, i \neq 1$, if $x'_i+b \le y$, let $x_i=2y-x'_i$;  if $x'_i < y< x'_i + b$, let $x_i =x'_i+2b$; if $x'_i-b<y<x'_i$, let $x_i =2y-x'_i+2b$; otherwise let $x_i=x'_i$. Observe that $|x_i-b-y|=\min\left(|x'_i-b-y|,|x'_i+b-y|\right)$ and hence $\text{cost}(x_i,y)=\text{cost}(x'_i,y)$ for every agent $i$. Similarily, it holds that $(x_i-b)-(x_1-b)\geq (x'_i-b)-(x'_1-b)$ and hence (since $x_1=x'_1$), $\text{cost}(x_i-b,x_1-b)\geq \text{cost}(x'_i-b,x'_1-b)$ for all agents $i$.

We will calculate an upper bound on the approximation ratio on instance $x'$. To do that, we will calculate an upper bound on the value of the ratio $SC_{\mathrm{\ref{mech}}}(\mathbf{x})/SC_{y}(\mathbf{x})$ on instance $\mathbf{x}$, where $SC_{\mathrm{\ref{mech}}}(\mathbf{x})$ is the social cost of Mechanism \ref{mech} on instance $\mathbf{x}$ and $SC_{y}(\mathbf{x})$ is the social cost admitted by $y$ on instance $\mathbf{x}$. By the way the instance was constructed, it holds that $SC_{\mathrm{\ref{mech}}}(\mathbf{x}) \geq SC_{\mathrm{\ref{mech}}}(\mathbf{x'})$ and $SC_{y}(\mathbf{x})=SC_{y}(\mathbf{x'})=SC_{\mathrm{opt}}(\mathbf{x'})$ and hence $SC_{\mathrm{\ref{mech}}}(\mathbf{x})/SC_{y}(\mathbf{x})$ is an upper bound on $SC_{\mathrm{\ref{mech}}}(\mathbf{x'})/SC_{\mathrm{opt}}(\mathbf{x'})$.

Let $d_i=(x_i-b)-y$, for $i=2,...,n$ and let $k = \sum_{i\neq 1}d_i$. Finally let $d = |y - (x_1+b)|$. We consider three cases. (See Figure~\ref{detercase}).\\

\textbf{Case 1:} $x_1+b < y$.


In this case, the social cost admitted by $y$ is $SC_{y}(\mathbf{x}) = nc+ k + d$, and the cost of Mechanism~\ref{mech} is $SC_{\mathrm{\ref{mech}}}(\mathbf{x})= nc + k + (n-1)(d+2b)$. The ratio is
\begin{eqnarray*}
\frac{SC_{\mathrm{\ref{mech}}}(\mathbf{x})}{SC_{y}(\mathbf{x})} = g(k,d,n) = \frac{k+(n-1)(d+2b)+nc}{k+d+nc}.
\end{eqnarray*}
By calculating the partial derivative of the ratio with respect to $k$, we have
\begin{eqnarray*}
\frac{\partial g(k,d,n)}{\partial k} &=& \frac{k+d+nc-k-(n-1)(d+2b)-nc}{(k+d+nc)^2} = \frac{d - (n-1)(d+2b)}{(k+d+nc)^2} < 0.
\end{eqnarray*}
When $k=0$, i.e. $x_i=y+b$ for all agents $i=2,\dots,n$, the ratio achieves the maximum value with respect to $k$. We plug in $k=0$ to the above formula and get
\begin{eqnarray*}
\frac{SC_{\mathrm{\ref{mech}}}(\mathbf{x})}{SC_{y}(\mathbf{x})} = g(d,n) &=& \frac{(n-1)(d+2b)+nc}{d+nc} = \frac{(n-1)d + 2(n-1)b+nc}{d+nc}.
\end{eqnarray*}
Again, we calculate the partial derivative of $g(d,n)$ with respect to $d$,
\begin{eqnarray*}
\frac{\partial g(d,n)}{\partial d} &=& \frac{(n-1)(d+nc) - (n-1)d - 2(n-1)b-nc}{(d+nc)^2} = \frac{c(n^2-4n+2)}{(d+nc)^2}.
\end{eqnarray*}
For sufficiently large $n$ ($n\geq 2b/c$), the above value is positive, so the ratio is maximized when $d$ approaches to infinity, which means that agent $x_1$ is positioned very far away from the rest of the agents coinciding on position $x_i$. Therefore the ratio goes to $n-1$. \\
					
\textbf{Case 2:} $x_1 < y \le x_1+b$, which means $0\le d<b$.


In this case, the social cost admitted by $y$ is $SC_{y}(\mathbf{x}) = k + d +nc$, and the cost of Mechanism~\ref{mech} is $SC_{\mathrm{\ref{mech}}}(\mathbf{x})=k+(n-1)(2b-d)+nc$. So the ratio is
\begin{eqnarray*}
\frac{SC_{\mathrm{\ref{mech}}}(\mathbf{x})}{SC_{y}(\mathbf{x})} = g(k,d,n) = \frac{k+(n-1)(2b-d)+nc}{k+d+nc}.
\end{eqnarray*}
Again, we calculate the partial derivative of $g(k,d,n)$ with respect to $k$,
\begin{eqnarray*}
\frac{\partial g(k,d,n)}{\partial k} &=& \frac{k+d+nc-k-(n-1)(2b-d)-nc}{(k+d+nc)^2} = \frac{d-(n-1)(2b-d)}{(k+d+nc)^2} .
\end{eqnarray*}
Since $0\le d<b$, $d < 2b-d$ and the above value is negative. To get the maximum value of the ratio, we plug in $k=0$ and the approximation ratio becomes
\begin{eqnarray*}
\frac{SC_{\mathrm{\ref{mech}}}(\mathbf{x})}{SC_{y}(\mathbf{x})} = g(d,n) = \frac{(n-1)(2b-d)+nc}{d+nc}.
\end{eqnarray*}
By calculating the partial derivative with respect to $d$, we get
\begin{eqnarray*}
\frac{\partial g(d,n)}{\partial d} &=& \frac{-(n-1)(d+nc)-(n-1)(2b-d)-nc}{(d+nc)^2}= \frac{-(n-2)nc-(n-1)2b}{(d+nc)^2} < 0.
\end{eqnarray*}
So the ratio is maximized when $d=0$, and the ratio is
$1+\frac{(2n-2)b + nc}{nc}$, which (for sufficiently large $n$) is smaller than the ratio in Case 1. \\

\textbf{Case 3:} $x_1-b<y \le x_1$, which means $b \le d <2b$.


In this case, the social cost admitted by $y$ is $SC_{y}(\mathbf{x}) = k + 2b-d + nc$, and cost of Mechanism~\ref{mech} is $SC_{\mathrm{\ref{mech}}}(\mathbf{x})=k+(n-1)(2b-d)+nc$. So the ratio is
\begin{eqnarray*}
\frac{SC_{\mathrm{\ref{mech}}}(\mathbf{x})}{SC_{y}(\mathbf{x})} = g(k,d,n) = \frac{k+(n-1)(2b-d)+nc}{k+2b-d+nc}.
\end{eqnarray*}
We first calculate partial derivative of the function $g$ with respect to $k$,
\begin{eqnarray*}
\frac{\partial g(k,d,n)}{\partial k} &=& \frac{k+2b-d+nc-k-(n-1)(2b-d)-nc}{(k+2b-d+nc)^2} = \frac{(2-n)(2b-d)}{(k+2b-d+nc)^2} < 0.
\end{eqnarray*}
The function is decreasing in $k$ so let's set $k=0$. The function $f$ becomes
\begin{eqnarray*}
\frac{SC_{\mathrm{\ref{mech}}}(\mathbf{x})}{SC_{y}(\mathbf{x})} = g(d,n) = \frac{(n-1)(2b-d)+nc}{2b-d+nc}.
\end{eqnarray*}
Now, we calculate partial derivative of $g$ with respect to $d$ and get
\begin{eqnarray*}
\frac{\partial g(d,n)}{\partial d} &=& \frac{-(n-1)(2b-d)-n(n-1)c +(n-1)(2b-d) +nc}{(2b-d+nc)^2} \le 0.
\end{eqnarray*}
So function $g(d,n)$ is decreasing in $d$. Let's set $d=b$, and the ratio goes to $\frac{(n-1)b+nc}{b+nc}$, which goes to $n-1$ as well.

In all, since $SC_{\mathrm{\ref{mech}}}(\mathbf{x})/SC_{y}(\mathbf{x}) \leq n-1$ the approximation ratio of Mechanism~\ref{mech} is at most $n-1$. The approximation ratio is exactly $n-1$ on any instance $\mathbf{x}=\langle x_1,x_2,...,x_n\rangle$ with $x_2=\dots=x_n$ and $x_1 \ll x_2$, i.e. when agent $1$ lies on the left, really far away from the other $n-1$ agents.
\hfill $\square$
\end{proof}

Next, we will prove a lower bound of $1+b/c$ on the approximation ratio of any anonymous, position invariant, and strategyproof mechanism, when the number of agents is even. We will start with the following lemma.

\begin{lemma}\label{twotomany}
Let $M^n$ be a strategyproof, anonymous and position invariant mechanism for $n$ agents, where $n$ is even. Then, for any location profile $\mathbf{x}=\langle x_1=\ldots=x_{n/2}, x_{n/2+1}=\ldots=x_n\rangle$, it holds that $M^n(\mathbf{x})= x_1-b$.
\end{lemma}

\begin{proof}
Let $M^2$ be the following mechanism for two agents: On input location profile $\langle x_1,x_2\rangle$, output $M^n(\mathbf{x'})$, where $\mathbf{x'}=\langle x'_1=\dots=x'_{n/2}, x'_{n/2+1}=\ldots=x'_n \rangle$, and $x'_1=x_1$ and $x'_{n/2+1}=x_2$. First, we claim that $M^2$ is strategyproof, anonymous and position invariant. If that is true, then by Theorem \ref{2unique}, $M^2$ is Mechanism \ref{mech} and the lemma follows.

First let $\mathbf{x}=\langle x_1,x_2 \rangle$, $\mathbf{\hat{x}}=\langle \hat{x}_1, \hat{x}_2 \rangle$ be any two position equivalent location profiles. Observe that the corresponding $n$-agent profiles $\mathbf{x'}$ and $\mathbf{\hat{x}'}$ obtained by placing $n/2$ agents on $x_1$ and $\hat{x_1}$ and $n/2$ agents on $x_2$ and $\hat{x_2}$ respectively are also position equivalent. Since $M^n$ is position invariant, it must hold that $M^n(\mathbf{x'})=M^n(\mathbf{\hat{x}'})$ and hence by construction of $M^2$, $M^2(\mathbf{x})=M^2(\mathbf{\hat{x}})$. Since $\mathbf{x}$ and $\mathbf{\hat{x}}$ where arbitrary, Mechanism $M^2$ is position invariant.

Similarly, let $\mathbf{x}=\langle x_1,x_2 \rangle$, $\mathbf{\hat{x}}=\langle \hat{x}_1, \hat{x}_2 \rangle$ be any two location profiles, such that $\mathbf{\hat{x}}$ is obtained by $\mathbf{x}$ by a permutation of the agents. The outcome of $M^n$ on the corresponding $n$-agent location profiles (since the number of agents placed on $x_1$ and $x_2$ is the same) is the same and by construction of $M^2$, $M^2(\mathbf{x})=M^2(\mathbf{\hat{x}})$ and since the profiles where arbitrary, the mechanism is anonymous.

Finally, for strategyproofness, start with a location profile $\mathbf{\hat{x}'}=\langle \hat{x}'_1,\hat{x}'_2 \rangle$ and let $\mathbf{x'}=\langle x'_1=\dots=x'_{n/2}, x'_{n/2+1}=\ldots=x'_n \rangle$ be the corresponding $n$-agent location profile. Let $y=M^n(\mathbf{x'})$ and let $\textrm{cost}(x',y)$ be the cost of agents $x'_1,\ldots,x'_{n/2}$ on $\mathbf{x'}$. For any $x_1$, let $\langle x_1, x'_2=\ldots=x'_{n/2}, x'_{n/2+1}=\ldots=x'_n \rangle$ be the resulting location profile. By strategyproofness of $M^n$, agent $x'_1$ can not decrease her cost by misreporting $x_1$ on profile $\mathbf{x'}$ and hence her cost on the new profile is at least $\textrm{cost}(x',y)$. Next, consider the location profile $\langle x_1=x_2,x'_3=\ldots=x'_{n/2}, x'_{n/2+1}=\ldots=x'_n \rangle$ and observe that by the same argument, the cost of agent $x'_2$ is not smaller on the new profile when compared to $\langle x_1, x'_2=\ldots=x'_{n/2}, x'_{n/2+1}=\ldots=x'_n \rangle$ and hence her cost is at least $\textrm{cost}(x',y)$. Continuing like this, we obtain the profile $\langle x_1=\ldots=x_{n/2},x'_{n/2+1}=\ldots=x'_n \rangle$ and by the same argument, the cost of agent $x'_{n/2}$ on this profile is at least $\textrm{cost}(x',y)$. The location profile $\langle x_1=\ldots=x_{n/2},x'_{n/2+1}=\ldots=x'_n \rangle$ corresponds to the $2$-agent location profile $\mathbf{\hat{x}}=\langle \hat{x}_1,\hat{x}'_2 \rangle$ and by construction of $M^2$, $\textrm{cost}(\hat{x}'_1,M^2(\mathbf{\hat{x}'})) \leq \textrm{cost}(\hat{x}'_1,M^2(\mathbf{\hat{x}}))$ and since the choice of $x_1$ (and hence the choice of $\hat{x_1}$) was arbitrary, Mechanism $M^2$ is strategyproof.  
\end{proof}

\begin{theorem}\label{detlowerbound}
When the number of agents is even, any strategyproof mechanism that satisfies position invariance and anonymity achieves an approximation ratio of at least $1+(b/c)$ for the social cost.
\end{theorem}

\begin{proof}
Let $M^n$ be a strategyproof, anonymous and position invariant mechanism and consider any location profile $\mathbf{x}=\langle x_1=\ldots=x_{n/2},x_{n/2}+1=\ldots=x_n$ with $x_{n/2+1}=x_1+2b$. By Lemma \ref{twotomany}, $M^n(\mathbf{x})= x_1-b$ and the social cost of $M^n$ is $nc+(n/2)2b$ while the social cost of the optimal allocation is only $nc$. The lower bound follows.   
\end{proof}

\subsection{Maximum Cost}
First, it is easy to see that the mechanism that outputs the location that minimizes the maximum cost is not strategyproof. On any instance $\langle x_1,x_2 \rangle$ with $x_1+b < x_2-b$ the optimal location of the facility is $(x_1+x_2)/2$. If agent $x_2$ misreports $x_2'=2x_2-2b-x_1$ then the location moves to $x_2-b$, minimizing her cost.

While the approximation ratio of Mechanism \ref{mech} for the social cost is not constant, for the maximum cost that is indeed the case. In fact, as we will prove, when the number of agents is even, Mechanism \ref{mech} actually achieves the best possible approximation ratio amongst strategyproof mechanisms. We start with the theorem about the approximation ratio of Mechanism \ref{mech}.

\begin{theorem}\label{deter-max-cost}
For $n \geq 3$, Mechanism \ref{mech} achieves an approximation ratio of $1+\frac{2b}{c}$ for the maximum cost.
\end{theorem}

\begin{proof}
Let $\mathbf{x}=\langle x_1,...,x_n\rangle$ be any instance. We consider three cases, depending on the distance between agents $x_1$ and $x_n$. \\

\textbf{Case 1:} $x_1+b \le x_n-b \Rightarrow x_n-x_1 \ge 2b$.

In this case, the cost of the optimal mechanism is $MC_{\mathrm{opt}}(\mathbf{x}) = (x_n -b) - \frac{x_1+x_n}{2} + c = \frac{x_n-x_1}{2}-b+c$, whereas the cost of Mechanism \ref{mech} is $MC_{\mathrm{\ref{mech}}}(\mathbf{x}) = x_n-x_1+c$. The approximation ratio is
\begin{eqnarray*}
\frac{MC_{\mathrm{\ref{mech}}}(\mathbf{x})}{MC_{\mathrm{opt}}(\mathbf{x})} = \frac{2(x_n-x_1+c)}{x_n-x_1-2b+2c} = 2+ \frac{4b-2c}{x_n-x_1-2b+2c} \le 1+\frac{2b}{c}.
\end{eqnarray*}
Hence, in this case, the approximation ratio is at most $1+\frac{2b}{c}$. For $x_n-b=x_1+b$ (the instance on which the right peak of the first agent and the left peak of the last agent coincide), the approximation ratio is exactly $1+\frac{2b}{c}$. \\

\textbf{Case 2:} $x_1 < x_n-b < x_1+b \Rightarrow c < x_n-x_1 < 2b$.

The cost of Mechanism \ref{mech} in this case is $MC_{\mathrm{\ref{mech}}}(\mathbf{x}) = x_n-b - (x_1-b)+c = x_n-x_1+c$, while the cost of the optimal mechanism is at least $c$. The approximation ratio is
\begin{eqnarray*}
\frac{MC_{\mathrm{\ref{mech}}}(\mathbf{x})}{MC_{\mathrm{opt}}(\mathbf{x})} \leq \frac{(x_n-x_1+c)}{c} < 1+\frac{2b}{c}
\end{eqnarray*}

\textbf{Case 3:} $x_1-b \leq x_n-b \leq x_1 \Rightarrow x_n-x_1 \leq b$.

The cost of Mechanism \ref{mech} is $MC_{\mathrm{\ref{mech}}}(\mathbf{x}) = x_n-x_1+c$ while the cost of the optimal mechanism is at least $c$. The approximation ratio is
\begin{eqnarray*}
\frac{MC_{\mathrm{\ref{mech}}}(\mathbf{x})}{MC_{\mathrm{opt}}(\mathbf{x})} \leq \frac{(x_n-x_1+c)}{c} \leq 1+\frac{b}{c}.
\end{eqnarray*}
Over the three cases, the worst approximation ratio is at most $1+\frac{2b}{c}$ and there is actually an instance of the problem with approximation ratio exactly $1+\frac{2b}{c}$, ensuring that the bound is tight.
\hfill $\square$
\end{proof}

The lower bound for the case when the number of agents is even follows.

\begin{corollary}
When the number of agents is even, any deterministic strategyproof mechanism that satisfies position invariance and anonymity achieves an approximation ratio of at least $1+\frac{2b}{c}$ for the maximum cost.
\end{corollary}

\begin{proof}
On instance $\mathbf{\hat{x}}$ of the proof of Theorem \ref{detlowerbound}, there are agents whose cost for any strategyproof, anonymous and position invariant mechanism is $2b+c$, while under the optimal mechanism it is only $c$. The lower bound on the approximation ratio follows.
\end{proof}

A worse lower bound that holds for any number of agents (and without using the position invariance property) is proved in the next theorem.

\begin{theorem}\label{unconditionallowerbound}
Any deterministic strategyproof mechanism achieves an approximation ratio of at least $2$ for the maximum cost.
\end{theorem}

\begin{proof}
Consider an instance $\mathbf{x}=\langle x_1,x_2 \rangle$ with $x_1+b < x_2-b$ (the instance can be extended to arbitrarily many agents by placing agents on positions $x_1$ and $x_2$ and all the arguments will still hold). The optimal location of the facility is $f_{\mathrm{opt}}(\mathbf{x}) = \frac{x_1+x_2}{2}$. Assume for contradiction that $M$ is a deterministic strategyproof mechanism with approximation ratio smaller than $2$. 

First, we argue that it can not be that $f_\mathrm{M}(\mathbf{x}) \in [x_2-b,\infty)$. Let $d = x_2-b-f_{\mathrm{opt}}(\mathbf{x})$. It holds that $MC_{\mathrm{opt}}(\mathbf{x})=c+d$. If it was $f_\mathrm{M}(\mathbf{x}) \in [x_2-b,\infty)$, then it would be that $MC_{\mathrm{M}}(\mathbf{x})\geq c+2d$ and the approximation ratio would be at least $2-\frac{c}{d+c}$ which goes to $2$ as $d$ grows to infinity (i.e. the agents are placed very far away from each other). Now, for Mechanism $M$ to achieve an approximation ratio smaller than $2$, it must be $f_{\mathrm{M}}(\mathbf{x}) \in [f_{\mathrm{opt}}(\mathbf{x}),x_2-b)$ (or symmetrically $f_{\mathrm{M}}(\mathbf{x}) \in (x_1+b,f_{\mathrm{opt}}(\mathbf{x})]$).

Now consider the instance $\mathbf{x'}=\langle x_1',x_2'\rangle$ with $x_1'=x_1$ and $x_2' = f_{\mathrm{M}}(\mathbf{x})+b$. On this instance, it must be either $f_{\mathrm{M}}(\mathbf{x'})=f_{\mathrm{M}}(\mathbf{x})$ or $f_{\mathrm{M}}(\mathbf{x'})=f_{\mathrm{M}}(\mathbf{x})+2b$ (the left or the right peak of agent $x_2'$), otherwise agent $x_2'$ could report $x_2$ and move the facility to $x_2'-b$, minimizing her cost and violating strategyproofness. For calculating the lower bound, we need the choice that admits the smaller of the two costs, i.e. $f_{\mathrm{M}}(\mathbf{x'})=f_{\mathrm{M}}(\mathbf{x})$. We calculate the approximation ratio on instance $\mathbf{x'}$.

The optimal choice for the facility is again $f_{\mathrm{opt}}(\mathbf{x'})=(x_1+x_2')/2$. Let $\lambda = (x_2'-b) - f_{\mathrm{opt}}(\mathbf{x'})$. The approximation ratio then is $2-\frac{c}{\lambda+c}$. We know that $\lambda \geq d/2$, so when $d$ grows to infinity as before, $\lambda$ also grows to infinity and the approximation ratio goes to $2$. This means that there exists an instance for which the approximation ratio of the mechanism is $2$, which gives us the lower bound.
\end{proof}

\section{Generalizations and conclusion}\label{generalizations}

As argued in the introduction, double-peaked preferences are often a very realistic model for facility location and our results initiate the discussion on such settings and shed some light on the capabilities and limitations of strategyproof mechanisms. We conclude with a discussion about an extension to the main model and some potential future directions.

\subsection{Non-symmetric peaks}
\begin{table}[t]
	\centering
	\caption{The results for the case when peaks are not required to be symmetric.}
	\label{nonsymmetrictable}
	\begin{center}
		\begin{tabular}{ |c  | c  c  |  }
			\hline
			&  \multicolumn{2}{|c|}{\textbf{Non-symmetric}}\\[3pt] 
			& Ratio & Lower  \\[2pt] \hline
			\textbf{Social cost}  &&  \\[2pt] 
			Deterministic &  $\Theta(n)$ & $1+\frac{b_1+b_2}{c}$ \\[2pt] 
			Randomized & $\Theta(n)$ & - \\[2pt]
			\hline
			\textbf{Maximum cost} & &   \\[3pt] 
			Deterministic & $\Theta(n)$ & $1+\frac{b_1+b_2}{c}$ \\[2pt]
			Randomized & $\Theta(n)$ & $3/2$ \\[2pt]
			\hline
		\end{tabular}
	\end{center}
\end{table}
Although the symmetric case is arguably the best analogue of the single-peaked preference setting, it could certainly make sense to consider a more general model, where the cost functions do not have the same slope in every interval and hence the peaks are not equidistant from the location of an agent. Let $b_1$ and $b_2$ be the distances from the left and the right peaks respectively. Clearly, all our lower bounds still hold, although one could potentially prove even stronger bounds by taking advantage of the more general setting.
The main observation is that Mechanism \ref{leftrightmedian} is no longer truthful-in-expectation, because its truthfulness depends heavily on the peaks being equidistant. On the other hand, mechanism \ref{mech} is still strategyproof and the approximation ratio bounds extend naturally. A summary of the results for the non-symmetric setting is depicted in Table \ref{nonsymmetrictable}.

\subsection{Future work}

Starting from randomized mechanisms, we would like to obtain lower bounds that are functions of $b$ and $c$, to see how well Mechanism \ref{leftrightmedian} fares in the general setting. For deterministic mechanisms, we would like to get a result that would clear up the picture. Characterizing strategyproof, anonymous and position invariant mechanisms would be ideal, but proving a lower bound that depends on $n$ on the ratio of such mechanisms (for the social cost) would also be quite helpful. The techniques used in our characterization for two agents and our lower bounds (available at the full version) seem to convey promising intuition for achieving such a task. Finally, it would be interesting to see if we can come up with a ``good'' randomized truthful-in-expectation mechanism for the extended model, when peaks are not assumed to be symmetric.

\bibliographystyle{plain}
\bibliography{ref}

\begin{thebibliography}{10}

\bibitem{alon2010strategyproof}
Noga Alon, Michal Feldman, Ariel~D Procaccia, and Moshe Tennenholtz.
\newblock Strategyproof approximation of the minimax on networks.
\newblock {\em Mathematics of Operations Research}, 35(3):513--526, 2010.

\bibitem{ashlagi2010mix}
I.~Ashlagi, F.~Fischer, I.~Kash, and Ariel~D. Procaccia.
\newblock Mix and match.
\newblock In {\em Proceedings of the 11th ACM conference on Electronic commerce
  (ACM-EC)}, pages 305--314. ACM, 2010.

\bibitem{black1986theory}
Duncan Black.
\newblock {\em The theory of committees and elections}.
\newblock Kluwer Academic Publishers, 1957 (reprint at 1986).

\bibitem{caragiannis2011improved}
I.~Caragiannis, A.~Filos-Ratsikas, and Ariel~D. Procaccia.
\newblock An improved 2-agent kidney exchange mechanism.
\newblock In {\em Proceedings of the 7th Workshop of Internet and Network
  Economics (WINE)}, pages 37--48. Springer, 2011.

\bibitem{cooper1999thestrategic}
Robert~D Cooter.
\newblock {\em The strategic constitution}.
\newblock Princeton University Press, 2002.

\bibitem{dokow2012mechanism}
Elad Dokow, Michal Feldman, Reshef Meir, and Ilan Nehama.
\newblock Mechanism design on discrete lines and cycles.
\newblock In {\em Proceedings of the 13th ACM Conference on Electronic Commerce
  (ACM-EC)}, pages 423--440, 2012.

\bibitem{dughmi2010truthful}
S.~Dughmi and A.~Gosh.
\newblock Truthful assignment without money.
\newblock In {\em Proceedings of the 11th ACM conference on Electronic commerce
  (ACM-EC)}, pages 325--334, 2010.

\bibitem{egan2013something}
Patrick~J Egan.
\newblock ''{D}o something'' politics and double-peaked policy preferences.
\newblock {\em Journal of Politics}, 76(2):333--349, 2013.

\bibitem{escoffier2011strategy}
B.~Escoffier, L.~Gourv{\`e}s, N.~Kim~Thang, F.~Pascual, and O.~Spanjaard.
\newblock Strategy-proof mechanisms for facility location games with many
  facilities.
\newblock {\em Algorithmic Decision Theory}, pages 67--81, 2011.

\bibitem{feigenbaum2013approximately}
Itai Feigenbaum, Jay Sethuraman, and Chun Ye.
\newblock Approximately optimal mechanisms for strategyproof facility location:
  Minimizing $ l\_p $ norm of costs.
\newblock {\em arXiv preprint arXiv:1305.2446}, 2013.

\bibitem{feldman2013strategyproof}
M.~Feldman and Y.~Wilf.
\newblock Strategyproof facility location and the least squares objective.
\newblock In {\em Proceedings of the 14th ACM conference on Electronic Commerce
  (ACM-EC)}, pages 873--890, 2013.

\bibitem{fotakis2010winner}
D.~Fotakis and C.~Tzamos.
\newblock Winner-imposing strategyproof mechanisms for multiple facility
  location games.
\newblock In {\em Proceeding of the 5th International Workshop of Internet and
  Network Economics (WINE)}, pages 234--245, 2010.

\bibitem{fotakis2012power}
D.~Fotakis and C.~Tzamos.
\newblock On the power of deterministic mechanisms for facility location games.
\newblock {\em arXiv preprint arXiv:1207.0935}, 2012.

\bibitem{guo2010strategyproof}
M.~Guo and V.~Conitzer.
\newblock Strategy-proof allocation of multiple items between two agents
  without payments or priors.
\newblock In {\em Ninth International Joint Conference on Autonomous Agents and
  Multi Agent Systems (AAMAS)}, volume~10, pages 881--888, 2010.

\bibitem{lu2010asymptotically}
Pinyan Lu, Xiaorui Sun, Yajun Wang, and Zeyuan~Allen Zhu.
\newblock Asymptotically optimal strategy-proof mechanisms for two-facility
  games.
\newblock In {\em Proceedings of the 11th ACM Conference on Electronic Commerce
  (ACM-EC)}, pages 315--324, 2010.

\bibitem{lu2009tighter}
Pinyan Lu, Yajun Wang, and Yuan Zhou.
\newblock Tighter bounds for facility games.
\newblock In {\em Proceeding of the 5th International Workshop of Internet and
  Network Economics (WINE)}, pages 137--148, 2009.

\bibitem{moulin1980strategy}
Herv{\'e} Moulin.
\newblock On strategy-proofness and single peakedness.
\newblock {\em Public Choice}, 35(4):437--455, 1980.

\bibitem{procaccia2009approximate}
Ariel~D Procaccia and Moshe Tennenholtz.
\newblock Approximate mechanism design without money.
\newblock In {\em Proceedings of the 10th ACM Conference on Electronic Commerce
  (ACM-EC)}, pages 177--186, 2009.

\bibitem{rosen2004public}
Harvey~S Rosen.
\newblock {\em Public Finance}, chapter 6 - Political Economy, pages 115--117.
\newblock Springer, 2004.

\bibitem{schummer2002strategy}
James Schummer and Rakesh~V Vohra.
\newblock Strategy-proof location on a network.
\newblock {\em Journal of Economic Theory}, 104(2):405--428, 2002.

\end{thebibliography}

\end{document}